\documentclass[tran,onecolumn,noversion,nonote]{cdcarticle}

\def\MODE{1}

\usepackage{cite}
\usepackage{graphicx}
\usepackage[cmex10]{amsmath}
\usepackage{amsthm,amssymb}
\usepackage{array}
\usepackage{enumerate}
\usepackage{arydshln}
\usepackage[hidelinks]{hyperref}

\usepackage{tikz}

\makeatletter
\def\@IEEElegacywarn#1#2{}
\makeatother

\newtheorem{thm}{Theorem}
\newtheorem{lem}[thm]{Lemma}

\newtheorem{defn}[thm]{Definition}
\newtheorem{cor}[thm]{Corollary}

\renewenvironment{proof}{\noindent{\bf Proof.}}{ \hfill ~\qed}
\newenvironment{proofe}[1]{\noindent{\bf #1}}{ \hfill ~\qed}

\def\qed{\rule[0pt]{5pt}{5pt}\par\medskip}

\newcommand{\T}{\rule{0pt}{2.6ex}}
\newcommand{\hlinet}{\hline\T}

\def\tp{\mathsf{T}}
\newcommand{\stsp}[4]{\left[\begin{array}{c|c}#1 & #2 \\ \hlinet #3 & #4\end{array}\right]}
\newcommand{\bmat}[1]{\begin{bmatrix}#1\end{bmatrix}}
\DeclareMathOperator{\are}{\mathrm{ARE}}
\DeclareMathOperator{\vecc}{\mathrm{vec}}
\DeclareMathOperator{\Lower}{\mathrm{lower}}

\DeclareMathOperator{\tr}{\mathrm{trace}}
\DeclareMathOperator*{\argmin}{\mathrm{arg\,min}}

\usepackage{colonequals}
\newcommand{\defeq}{\colonequals}

\newcommand{\R}{\mathbb{R}} 
\newcommand{\Z}{\mathbb{Z}} 
\newcommand{\C}{\mathbb{C}} 
\newcommand{\Rp}{\mathcal{R}_p} 
\newcommand{\Htwo}{\mathcal{H}_2}
\newcommand{\Ltwo}{\mathcal{L}_2}
\newcommand{\RHinf}{\mathcal{RH}_\infty} 
\newcommand{\RHtwo}{\mathcal{RH}_2} 
\newcommand{\RLtwo}{\mathcal{RL}_2} 
\renewcommand{\i}{\mathrm{i}}
\newcommand{\cl}{{c\ell}}

\let\bl\bigl

\let\bbbl\biggl
\let\bbbbl\Biggl
\let\br\bigr

\let\bbbr\biggr
\let\bbbbr\Biggr

\newcommand{\norm}[1]{\lVert{#1}\rVert}
\newcommand{\normm}[1]{\bl\lVert{#1}\br\rVert}

\newcommand{\normmmmm}[1]{\bbbbl\lVert{#1}\bbbbr\rVert}

\newcommand{\vertt}{\big\vert}

\newcommand{\snegskip}{\vspace{-1mm}}
\newcommand{\negskip}{\vspace{-2mm}}

\def\note#1{}

\begin{document}
\title{Optimal Control of Two-Player Systems\\ with Output Feedback}

\author{Laurent~Lessard\if\MODE2~and\else\and\fi~Sanjay~Lall%
\if\MODE1\else%
\thanks{L.~Lessard is with the Department of Mechanical Engineering at the University of California, Berkeley, CA 94720, USA. \texttt{lessard@berkeley.edu}}%
\thanks{S.~Lall is with the Department of Electrical Engineering, and
the Department of Aeronautics and Astronautics at Stanford
University, Stanford, CA 94305, USA. \texttt{lall@stanford.edu}}%
\thanks{This work was partly done while the first author was an LCCC postdoc at Lund University in Lund, Sweden.}%
\thanks{The second author is partially supported by the U.S. Air Force Office of Scientific Research (AFOSR) under grant number MURI FA9550-10-1-0573.}\fi}

\note{To Appear, IEEE Transactions on Automatic Control, 2015}

\ifCLASSOPTIONpeerreview
\markboth{IEEE Transactions on Automatic Control}{}
\else
\markboth{IEEE Transactions on Automatic Control}%
{Lessard \MakeLowercase{\textit{et al.}}: Optimal 
  Control of Two-Player Systems with Output Feedback}
\fi

\maketitle


\begin{abstract}
In this article, we consider a fundamental decentralized optimal control problem, which we call the two-player problem. Two subsystems are interconnected in a nested information pattern, and output feedback controllers must be designed for each subsystem. Several special cases of this architecture have previously been solved, such as the state-feedback case or the case where the dynamics of both systems are decoupled. In this paper, we present a detailed solution to the general case. The structure of the optimal decentralized controller is reminiscent of that of the optimal centralized controller; each player must estimate the state of the system given their available information and apply static control policies to these estimates to compute the optimal controller. The previously solved cases benefit from a separation between estimation and control that allows the associated gains to be computed separately.  This feature is not present in general, and some of the gains must be solved for simultaneously. We show that computing the required coupled estimation and control gains amounts to solving a small system of linear equations.
\end{abstract}

\if\MODE2
\fi

\IEEEpeerreviewmaketitle

\section{Introduction}

\label{sec:intro}

\IEEEPARstart{M}{any} large-scale systems such as the internet, power
grids, or teams of autonomous vehicles, can be viewed as a network of
interconnected subsystems. A common feature of these applications is
that subsystems must make control decisions with limited
information. The hope is that despite the decentralized nature of the
system, global performance criteria can be optimized. 
In this paper, we consider a particular class of 
decentralized control problems, illustrated in Figure~\ref{fig:dec},
and develop the optimal linear controller in this framework.\snegskip

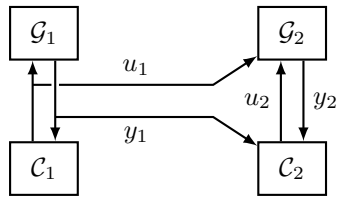
\begin{figure}[ht]
\centering
\tikzstyle{bigblock}=[draw,rectangle,minimum width=0.9cm,minimum height=0.7cm]
\def\sepdist{0.2}
\def\hordist{3.3}
\def\sep{0.15}
\def\offsep{0.75}
\def\xoff{2.4}
\def\noisedist{0.4}
\def\kl{0.1pt}
\def\koff{0.05}

\begin{tikzpicture}[thick,auto,>=latex,node distance=1.8cm,black]
\node [bigblock](P1){$\mathcal{G}_1$};
\node [bigblock,below of=P1](K1){$\mathcal{C}_1$};
\draw [->] (P1.south)+(\sep,0) node(tt1){} -- node{} (tt1 |- K1.north);
\draw [->] (K1.north)+(-\sep,0) node(tt2){} -- node{} (tt2 |- P1.south);
\path (P1) +(\hordist,0) node [bigblock](P2){$\mathcal{G}_2$};
\node [bigblock,below of=P2](K2){$\mathcal{C}_2$};
\draw [->] (P2.south)+(\sep,0) node(tt3){} -- node{$y_2$} (tt3 |- K2.north);
\draw [->] (K2.north)+(-\sep,0) node(tt4){} -- node{$u_2$} (tt4 |- P2.south);
\draw [->] (tt1)+(0,-\offsep) -- +(\koff,-\offsep) -- node[swap]{$y_1$} +(\xoff-2*\sep,-\offsep) -- (K2);
\draw [-] (tt2)+(0,\offsep) -- +(2*\sep-\koff,\offsep);
\draw [->] (tt2)+(2*\sep+\koff,\offsep) -- node{$u_1$} +(\xoff,\offsep) -- (P2);
\end{tikzpicture}
\caption{Decentralized interconnection\label{fig:dec}}
\end{figure}

Figure~\ref{fig:dec} shows plants $\mathcal{G}_1$ and
$\mathcal{G}_2$, with associated controllers $\mathcal{C}_1$
and~$\mathcal{C}_2$.  Controller~$\mathcal{C}_1$ receives measurements
only from $\mathcal{G}_1$, whereas $\mathcal{C}_2$ receives
measurements from both $\mathcal{G}_1$ and $\mathcal{G}_2$.  The
actions of $\mathcal{C}_1$ affect both $\mathcal{G}_1$ and
$\mathcal{G}_2$, whereas the actions of $\mathcal{C}_2$ only affect
$\mathcal{G}_2$.
In other words, information may only flow from
left to right. 
{\color{black}
We assume explicitly that the signal $u_1$ entering $\mathcal{G}_1$ is the same
as the signal $u_1$ entering $\mathcal{G}_2$. A different modeling assumption would be to assume that these signals could be affected by separate noise. Similarly, we assume that the measurement $y_1$ received by $\mathcal{C}_{1}$ is the same as that received by $\mathcal{C}_2$.
We give a precise definition of our chosen notion of stability in Section~\ref{ssec:probstatement} and we further discuss our choice of stabilization framework in Section~\ref{sec:stabilization}.}

The feedback interconnection of Figure~\ref{fig:dec} may be represented using the linear fractional transformation shown in Figure~\ref{fig:4block}.

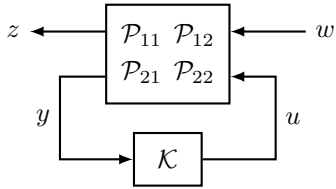
\begin{figure}[ht]
\centering
\tikzstyle{block}=[draw,rectangle,inner sep=2mm,minimum width=0.9cm,minimum height=0.7cm]
\begin{tikzpicture}[thick,auto,>=latex,node distance=1.4cm]
\node [block](P){$\begin{matrix}
    \mathcal{P}_{11}\rule[-1.3ex]{-1.3ex}{0pt}& \mathcal{P}_{12}\\
    \mathcal{P}_{21}\rule{-1.3ex}{0pt} & \mathcal{P}_{22}\end{matrix}$};
\node [block,below of=P](K){$\mathcal{K}$};
\draw [<-] (P.east)+(0,-0.3) -- +(0.6,-0.3) |- node[pos=0.25]{$u$} (K);
\draw [<-] (P.east)+(0,0.3) -- +(1,0.3) node [anchor=west]{$w$};
\draw [->] (P.west)+(0,-0.3) -- +(-0.6,-0.3) |- node[swap,pos=0.25]{$y$} (K);
\draw [->] (P.west)+(0,0.3) -- +(-1,0.3) node [anchor=east]{$z$};
\end{tikzpicture}
\caption{General four-block plant with controller in feedback \label{fig:4block}}
\end{figure}

The information flow in this problem leads to a sparsity structure for both~$\mathcal{P}$ and~$\mathcal{K}$ in Figure~\ref{fig:4block}. Specifically, we find the following
block-lower-triangular sparsity pattern
\begin{equation}\label{Ksparsity}
\mathcal{P}_{22} = \bmat{\mathcal{G}_{11} & 0 \\ \mathcal{G}_{21} & \mathcal{G}_{22}}
\quad\text{and}\quad
\mathcal{K} = \bmat{\mathcal{K}_{11} & 0 \\ \mathcal{K}_{21} & \mathcal{K}_{22}}
\end{equation}
while $\mathcal{P}_{11}$, $\mathcal{P}_{12}$, $\mathcal{P}_{21}$ are full in general.
We assume $\mathcal{P}$ and~$\mathcal{K}$ are finite-dimensional
continuous-time linear time-invariant systems.  The goal is to find an~$\mathcal{H}_2$-optimal
controller $\mathcal{K}$ subject to the constraint~\eqref{Ksparsity}.  
We paraphrase our main result, found in Theorem~\ref{thm:main}.
Consider the state-space dynamics
\[
\begin{alignedat}{2}
\bmat{\dot{x}_1 \\ \dot{x}_2}
&= \bmat{A_{11} & 0 \\ A_{21} & A_{22}} &&\bmat{x_1 \\ x_2} 
 + \bmat{B_{11} & 0 \\ B_{21} & B_{22}} \bmat{u_1 \\ u_2} + w \\
\bmat{y_1 \\ y_2}
&= \bmat{C_{11} & 0 \\ C_{21} & C_{22}} &&\bmat{x_1 \\ x_2} + v
\end{alignedat}
\]
where $w$ and $v$ are Gaussian white noise. The objective is to minimize the
quadratic cost of the standard LQG problem, subject to the constraint that
\[
\begin{aligned}
& \text{Player $1$ measures $y_1$ and chooses $u_1$} \\
& \text{Player $2$ measures $y_1,y_2$ and chooses $u_2$} 
\end{aligned}
\]
Here we borrow the terminology from game theory, and envisage
separate \emph{players} performing the actions $u_1$ and $u_2$.  We will
show that the optimal controller has the form
\[
 u  = K x_{|y_1,u_\zeta} + 
  \bmat{0 & 0 \\ H & J} 
  \bl(
  x_{|y,u}
  - x_{|y_1,u_\zeta}
  \br)
\]
where $x_{|y_1,u_\zeta}$ denotes the optimal estimate of $x$ using the
information available to Player~1, and $ x_{|y,u}$ is the optimal
estimate of $x$ using the information available Player~2. The matrices
$K$, $J$, and $H$ are determined from the solutions to Riccati and
Sylvester equations. The precise meaning of \emph{optimal estimate}
and $u_\zeta$ will be explained in Section~\ref{sec:struct}.  Our
results therefore provide a type of separation principle for such
problems. Questions of separation are central to decentralized
control, and very little is known in the general case. Our results
therefore also shed some light on this important issue.  Even though
the controller is expressed in terms of optimal estimators, these do
not all evolve according to a standard Kalman filter (also known as
the Kalman--Bucy filter), since Player~1, for example, does not know
$u_2$.

The main contribution of this paper is an explicit state-space formula
for the optimal controller, which was not previously known.  The
realization we find is generically minimal, and computing it is of
comparable computational complexity to computing the optimal
centralized controller. The solution also gives an intuitive
interpretation for the states of the optimal controller. 

The paper is organized as follows. In the remainder of the
introduction, we give a brief history of decentralized control and the
two-player problem in particular. Then, we cover some background
mathematics and give a formal statement of the two-player optimal
control problem. In Section~\ref{sec:stabilization}, we characterize
all structured stabilizing controllers, and show that the two-player
control problem can be expressed as an equivalent structured
model-matching problem that is convex. In Section~\ref{sec:main}, we
state our main result which contains the optimal controller formulae.
We follow up with a discussion of the structure and interpretation of
this controller in Section~\ref{sec:struct}. The subsequent
Section~\ref{sec:main_proof} gives a detailed proof of the main
result.  Finally, we conclude in Section~\ref{sec:conclusion}.


\subsection{Prior work}\label{ssec:prior work}

If we consider the problem of Section~\ref{sec:intro} but remove the
structural constraint~\eqref{Ksparsity}, the problem becomes the
well-studied classical $\mathcal{H}_2$ synthesis, solved for example
in~\cite{zdg}. The optimal controller is then linear and has
as many states as the plant.

The presence of structural constraints greatly complicates the
problem, and the resulting decentralized problem has been outstanding
since the seminal paper by Witsenhausen~\cite{witsenhausen} in 1968.
Witsenhausen posed a related problem for which a nonlinear controller
strictly outperforms all linear policies. Not all structural
constraints lead to nonlinear optimal controllers, however. For a
broad class of decentralized control problems there exists a linear
optimal policy~\cite{hochu}. Of recent interest have been classes of
problems for which finding the optimal linear controller amounts to
solving a convex optimization
problem~\cite{rotkowitz06,voulgaris00,voulgaris_stabilization}. The
two-player problem considered in the present work is one such case.

Despite the benefit of convexity, the search space is
infinite-dimensional since we must optimize over transfer functions.
Several numerical and analytical approaches for addressing
decentralized optimal control exist,
including~\cite{rantzer06,scherer02,zelazo09,voulgaris_stabilization}.
One particularly relevant numerical approach is to use vectorization,
which converts the decentralized problem into an equivalent
centralized problem~\cite{rotkowitz_2006a}. This conversion process
results in a dramatic growth in state dimension, and so the method is
 computationally intensive and only feasible for small
problems.  However, it does provide important insight into the
problem. Namely, it proves that there exists an optimal linear controller for the two-player problem considered herein that is rational.
These results are discussed in Section~\ref{sec:origins_a}.

A drawback of such numerical approaches is that they do not provide an
intuitive explanation for what the controller is doing; there is no
physical interpretation for the states of the controller. In the
centralized case, we have such an interpretation. Specifically, the
controller consists of a Kalman filter whose states are estimates of
the states of the plant together with a static control gain that corresponds
to the solution of an LQR problem.  Recent structural
results~\cite{sachin,mahajan} reveal that for general classes of
delayed-sharing information patterns, the optimal control policy
depends on a summary of the information available and a dynamic
programming approach may be used to compute it. There are general results also in the case where not all information is eventually available to all players~\cite{star}. However, these dynamic programming results do not appear to easily translate to state-space formulae for linear systems.

For linear systems in which each player eventually has access to all information, explicit formulae were found in~\cite{lamperski_delayed_recent}. A simple case with varying communication delay was treated in~\cite{matni}. 
Cases where plant data are only locally available have also been studied~\cite{mishra,farokhi}.

Certain special cases of the two-player problem have been solved
explicitly and clean physical interpretations have been found for the
states of the optimal controller.  Most notably, the state-feedback
case admits an explicit state-space solution using a spectral
factorization approach~\cite{swigart10}. This approach was also used
to address a case with partial output feedback, in which there is
output feedback for one player and state feedback for the
other~\cite{swigart_partial}.  The work of~\cite{shah10} also provided
a solution to the state-feedback case using the M\"obius transform
associated with the underlying poset. Certain special cases were also
solved in~\cite{jonghan}, which gave a method for splitting
decentralized optimal control problems into multiple centralized
problems. This splitting approach addresses cases other than
state-feedback, including partial output-feedback, and dynamically
decoupled problems.

In this article, we address the general two-player output-feedback
problem.  Our approach is perhaps closest technically to the work
of~\cite{swigart10} using spectral factorization, but uses the
factorization to split the problem in a different way, allowing a
solution of the general output-feedback problem. We also provide a
meaningful interpretation of the states of the optimal controller.
This paper is a substantially more general version of the invited
paper~\cite{allerton} and the conference
paper~\cite{lessard2012optimal}, where the special case of stable
systems was considered. We also mention the related work~\cite{cdc}
which addresses star-shaped systems in the stable case.


\section{Preliminaries}

We use $\Z_+$ to denote the nonnegative integers.
The imaginary unit is $\i$, and we denote the imaginary
axis by $\i\R$. A square matrix $A\in\R^{n\times n}$ is Hurwitz
if all of its eigenvalues have a strictly negative real part. The set
$\Ltwo(\i\R)$, or simply~$\Ltwo$, is a Hilbert space of Lebesgue
measurable matrix-valued functions $\mathcal{F}: \i\R\to\C^{m\times
  n}$ with the inner product
\[
\langle \mathcal{F},\mathcal{G} \rangle = \frac{1}{2\pi}
\int_{-\infty}^\infty \tr\bl(
 \mathcal{F}^*(\i\omega)\mathcal{G}(\i\omega)\br)
\,\mathrm{d}\omega
\]
such that the inner product induced norm $\|\mathcal{F}\|_2 = \langle
\mathcal{F},\mathcal{F} \rangle^{1/2}$ is bounded. We will sometimes
write $\Ltwo^{m\times n}$ to be explicit about the matrix dimensions.
As is standard, $\Htwo$ is a closed subspace of $\Ltwo$ with matrix
functions analytic in the open right-half plane. $\Htwo^\perp$ is the
orthogonal complement of~$\Htwo$ in~$\Ltwo$. We write~$\Rp$ to denote
the set of proper real rational transfer functions. We also
use~$\mathcal{R}$ as a prefix to modify other sets to indicate the
restriction to real rational functions. So~$\RLtwo$ is the set of
strictly proper rational transfer functions with no poles on the
imaginary axis, and~$\RHtwo$ is the stable subspace of~$\RLtwo$.  The
set of stable proper transfer functions is denoted~$\RHinf$. For the remainder of this paper, whenever we write $\norm{\mathcal{G}}_2$, it will always be the case that~$\mathcal{G} \in \RHtwo$.
Every $\mathcal{G} \in \Rp$ has a state-space realization
	\if\MODE2
	\[
	\mathcal{G} = \stsp{A}{B}{C}{D} = D+C(sI-A)^{-1}B
	\]
	with\negskip
	\[
	\mathcal{G}^* = \stsp{-A^\tp }{C^\tp }{-\!B^\tp }{D^\tp }
	\]
\else
	\[
	\mathcal{G} = \stsp{A}{B}{C}{D} = D+C(sI-A)^{-1}B
	\qquad\text{with}\qquad
	\mathcal{G}^* = \stsp{-A^\tp }{C^\tp }{-\!B^\tp }{D^\tp }
	\]
\fi
where $\mathcal{G}^*$ is the conjugate transpose of $\mathcal{G}$.  If
this realization is chosen to be stabilizable and detectable,
then~$\mathcal{G}\in\RHinf$ if and only if~$A$ is Hurwitz,
and~$\mathcal{G}\in\RHtwo$ if and only if~$A$ is Hurwitz and $D=0$.
For a thorough introduction to these topics, see~\cite{zdg}.

The plant $\mathcal{P} \in \Rp$ maps exogenous inputs~$w$ and actuator
inputs~$u$ to regulated outputs~$z$ and measurements~$y$. We seek a control law $u = \mathcal{K}y$ where
$\mathcal{K}\in\Rp$ so that the closed-loop map has some desirable
properties. The closed-loop map is depicted in
Figure~\ref{fig:4block}, and corresponds to the equations
\begin{equation}
  \label{plantP}
  \begin{aligned}
    \bmat{z \\ y} &= \bmat{ \mathcal{P}_{11} & \mathcal{P}_{12} \\
      \mathcal{P}_{21} & \mathcal{P}_{22} } \bmat{w \\ u}, & 
    u &= \mathcal{K} y
  \end{aligned}
\end{equation}
The closed-loop map from $w$ to $z$ is given by the lower
  linear fractional transform (LFT) defined by 
$ F_\ell(\mathcal{P},\mathcal{K}) \defeq
\mathcal{P}_{11} + \mathcal{P}_{12}
\mathcal{K}\left(I-\mathcal{P}_{22}\mathcal{K}\right)^{-1}
\mathcal{P}_{21} $.  If $\mathcal{K} = F_\ell(\mathcal{J},\mathcal{Q})$
and $\mathcal{J}$ has a proper inverse, the LFT may be inverted
according to $\mathcal{Q} = F_u(\mathcal{J}^{-1},\mathcal{K})$, where
$F_u$ denotes the upper LFT, defined as $
F_u(\mathcal{M},\mathcal{K}) \defeq \mathcal{M}_{22} + \mathcal{M}_{21}
\mathcal{K}\left(I-\mathcal{M}_{11}\mathcal{K}\right)^{-1}
\mathcal{M}_{12} $.

We will express the results of this paper in terms of the solutions to
algebraic Riccati equations (ARE) and recall here the basic facts. If
$D^\tp D > 0$, then the following are equivalent.

\begin{enumerate}[(i)]
\if\MODE2\itemsep=2mm\fi
\item There exists $X\in\R^{n \times n}$ such that\snegskip
\if\MODE2
	\begin{multline*}\label{eqn:are_basic}
	A^\tp X + XA +C^\tp C\\
	-(X B + C^\tp D)(D^\tp D)^{-1}(B^\tp X + D^\tp C)=0
	\end{multline*}
\else
	\begin{equation*}\label{eqn:are_basic}
	A^\tp X + XA +C^\tp C-(X B + C^\tp D)(D^\tp D)^{-1}(B^\tp X + D^\tp C)=0 
	\end{equation*}
\fi
and  $A-B(D^\tp D)^{-1}(B^\tp X + D^\tp C)$ is Hurwitz

\item $(A,B)$ is stabilizable and $\bmat{A -\i\omega I & B \\ C & D}$
  has full
  \if\MODE2\\\fi
  column rank for all $\omega \in \R$.

\end{enumerate}
Under these conditions, there exists a unique $X\in\R^{n\times n}$
satisfying~(i). This $X$ is symmetric and positive semidefinite, and is
called the \emph{stabilizing solution} of the ARE.  As a short form we will write $ (X,K) = \are(A,B,C,D) $
where $K = -(D^\tp D)^{-1}(B^\tp X + D^\tp C)$ is the associated
gain. 

\subsection{Block-triangular matrices}

Due to the triangular structure of our problem, we
also require notation to describe sets of block-lower-triangular
matrices. To this end, suppose $R$ is a commutative ring, $m,n\in\Z_+^2$ and
$m_i, n_i \geq 0$.  We define $\Lower(R,m,n)$ to be the set of
block-lower-triangular matrices with elements in $R$ partitioned
according to the index sets $m$ and $n$. That is,
$X\in \Lower(R,m,n)$ if and only if
\[
X = \bmat{X_{11} & 0 \\ X_{21} & X_{22}}
\quad\text{where}\quad
X_{ij} \in R^{m_i \times n_j}
\]
We sometimes omit the indices and simply write $\Lower(R)$.  We
also define the matrices~$E_1 = \bmat{I & 0}^\tp$ and $E_2 =
\bmat{0 & I}^\tp$, with dimensions to be inferred from context.  For
example, if we write $XE_1$ where $X\in\Lower(\R,m,n)$, then we mean $E_1 \in
\R^{(n_1+n_2)\times n_1}$. When writing $A\in \Lower(R)$, we allow for
the possibility that some of the blocks may be empty. For example, if
$m_1=0$  then we encounter the trivial case where
$\Lower(R,m,n)=R^{m_2\times (n_1+n_2)}$.

There is a correspondence between proper transfer functions
$\mathcal{G}\in \Rp$ and state-space realizations $(A,B,C,D)$. The
next result shows that a specialized form of this correspondence
exists when $\mathcal{G}\in\Lower(\Rp)$.

\begin{lem}
  \label{lem:triangular_realization}
  Suppose $\mathcal{G}\in \Lower(\mathcal{R}_p,k,m)$, and a
  realization for $\mathcal{G}$ is given by $(A,B,C,D)$. Then there
  exists $n\in\Z_+^2$ and an invertible matrix $T$ such that
  \begin{align*}
    TAT^{-1} &\in \Lower(\R,n,n)  &  TB &\in\Lower(\R,n,m) \\
    CT^{-1}&\in\Lower(\R,k,n) &  D&\in\Lower(\R,k,m)
  \end{align*}
\end{lem}

\begin{proof}
  Partition the state-space matrices according to the partition
  imposed by $k$ and~$m$.
  \[
  \mathcal{G} = \bmat{\mathcal{G}_{11} & 0 \\ \mathcal{G}_{21} & \mathcal{G}_{22}} =
  \left[\begin{array}{c|cc}
      A & B_1 & B_2 \\ \hlinet
      C_1 & D_{11} & 0 \\
      C_2 & D_{21} & D_{22}
    \end{array}\right]
  \]
  Note that we immediately have $D\in \Lower(\mathbb{R},k,m)$. However,
  $A$, $B$, and $C$ need not have the desired structure. If $k_1=0$ or $m_2=0$, then $\mathcal{G}_{12}$ is empty, the sparsity pattern is trivial,
  and any realization $(A,B,C,D)$ will do. Suppose $\mathcal{G}_{12}$ is non-empty and let $T$ be the
  matrix that transforms $\mathcal{G}_{12}$ into Kalman canonical form. 
  There are typically four blocks in such a decomposition, but since
  $\mathcal{G}_{12} = 0$, there can be no modes that are both
  controllable and observable. Apply the same $T$-transformation to
  $\mathcal{G}$, and obtain the realization
  \begin{equation}\label{eq:G_real_tri}
  \mathcal{G} = \left[\begin{array}{ccc|cc}
      A_{\bar c o} & 0 & 0 & B_{11} & 0\\
      A_{21} & A_{\bar c \bar o} & 0 & B_{21} & 0 \\
      A_{31} & A_{32} & A_{c\bar o} & B_{31} & B_{c \bar o} \\ \hlinet
      C_{\bar c o} & 0 & 0 & D_{11} & 0\\
      C_{21} & C_{22} & C_{23} & D_{21} & D_{22}
    \end{array}\right]
  \end{equation}
  This realization has the desired sparsity pattern, and we notice
  that there may be many admissible index sets $n$. For example, the
  modes $A_{\bar c \bar o}$ can be included into either the $A_{11}$
  block or the $A_{22}$ block. Note that
  \eqref{eq:G_real_tri} is an admissible realization even if some of
  the diagonal blocks of $A$ are empty.
\end{proof}
Lemma~\ref{lem:triangular_realization} also holds for more general sparsity patterns~\cite{realizability}.


\subsection{Problem statement}\label{ssec:probstatement}

We seek controllers $\mathcal{K}\in \Rp$ such that when they are
connected in feedback to a plant~\hbox{$\mathcal{P}\in\Rp$} as in
Figure~\ref{fig:4block}, the plant is stabilized.  Consider the
interconnection in Figure~\ref{fig:4block_stab}.
We say that $\mathcal{K}$ stabilizes $\mathcal{P}$ if the transfer function
$(w,u_1,u_2)\to (z,y_1,y_2)$ is well-posed and
stable. Well-posedness amounts to
$I-\mathcal{P}_{22}(\infty)\mathcal{K}(\infty)$ being nonsingular.\snegskip

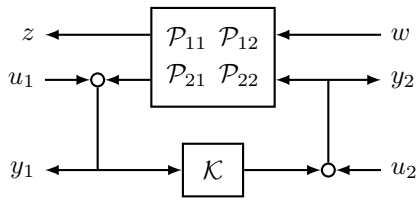
\begin{figure}[ht]
\centering
\tikzstyle{block}=[draw,rectangle,inner sep=2mm,minimum width=0.8cm,minimum height=0.7cm]
\tikzstyle{plus}=[draw,circle,inner sep=0.4ex]
\begin{tikzpicture}[thick,auto,>=latex,node distance=1.5cm]
\node [block](P){$\begin{matrix}\mathcal{P}_{11}\rule[-1.3ex]{-1.3ex}{0pt}& \mathcal{P}_{12}\\
\mathcal{P}_{21}\rule{-1.3ex}{0pt} & \mathcal{P}_{22}\end{matrix}$};
\node [block,below of=P](K){$\mathcal{K}$};
\draw [->] (P.west)+(0,0.3) -- +(-1.4,0.3) node [anchor=east](z){$z$};
\draw [->] (P.west)+(0,-0.3) -- +(-0.6,-0.3) node[plus,anchor=east](p1){};
\draw [->] (P.west)+(-1.4,-0.3) node [anchor=east](u1){$u_1$} -- (p1);
\draw [<->] (K.west) -- (u1.east |- K) node [anchor=east](y1){$y_1$};
\draw (p1) -- (p1 |- K);
\draw [<-] (P.east)+(0,0.3) -- +(1.4,0.3) node [anchor=west](w){$w$};
\path (P.east)+(0.6,-0.3) node (tmp){};
\draw [->] (K) -- (tmp |- K) node[plus,anchor=west](p2){};
\draw [<-] (p2) -- (w.west |- K) node [anchor=west](u2){$u_2$};
\draw [<->] (P.east)+(0,-0.3) -- +(1.4,-0.3) node [anchor=west](y2){$y_2$};
\draw (p2) -- (p2 |- y2);
\end{tikzpicture}
\caption{Feedback loop with additional inputs and outputs for
  analysis of stabilization of two-input two-output systems\label{fig:4block_stab}}
\end{figure}

The problem addressed in this article may be formally stated as
follows. Suppose $\mathcal{P}\in\Rp$ is given. Further suppose
that~$\mathcal{P}_{22} \in \Lower(\Rp,k,m)$.  The two-player problem
is
\begin{equation} 
  \label{opt:2p_output_feedback}
  \begin{aligned}
    \textup{minimize}\qquad
    & \big\|
      \mathcal{P}_{11} + \mathcal{P}_{12} 
      \mathcal{K}(I-\mathcal{P}_{22}\mathcal{K})^{-1} \mathcal{P}_{21} 
     \big\|_2 \\
    \textup{subject to}\qquad& \mathcal{K} \in \Lower(\Rp,m,k) \\
    & \mathcal{K}\text{ stabilizes }\mathcal{P}
\end{aligned}
\end{equation}
We will also make some additional standard assumptions. We will assume that $\mathcal{P}_{11}$ and $\mathcal{P}_{22}$ are
strictly proper, which ensures that the interconnection of Figure~\ref{fig:4block_stab} is always well-posed, and we will make some  technical assumptions
on $\mathcal{P}_{12}$ and $\mathcal{P}_{21}$ in order to guarantee the
existence and uniqueness of the optimal controller. The first step in our
solution to the two-player problem~\eqref{opt:2p_output_feedback} is to
deal with the stabilization constraint. This is the topic of the next section.


\section{Stabilization of triangular systems}

\label{sec:stabilization}

{\color{black}
Closed-loop stability for decentralized systems is a subtle issue. In
the centralized case, the core idea of pole-zero cancellation is
ancient, and this was beautifully extended to multivariable system in
the Desoer--Chan theory of closed-loop stability~\cite{desoer_chan}, where the
interconnection of the plant and controller in Figure~\ref{fig:4block_stab} is considered closed-loop stable if and only if the transfer function
$(w,u_1,u_2)\to (z,y_1,y_2)$ is well-posed and
stable.
Other modeling assumptions are possible, where one either
includes different plant uncertainty or different injected and
output signals. These assumptions would lead to different definitions
of stability.
For SISO systems, these two notions were shown to be
equivalent, and thus robustness to noise added to communication
signals between the plant and the controller is equivalent to this
type of robustness to plant modeling error.  Several works have proposed extensions of these ideas to decentralized systems, including~\cite{voulgaris_architectures,vamsi_elia}.  However, it is as yet poorly understood exactly what the correspondence is between plant uncertainty and signal uncertainty, and also it remains unclear which the relevant definition of decentralized stability is in practice. We therefore stick to the
well-established notion of closed-loop stability used for centralized
control systems.
}

In this section, we provide a state-space characterization of
stabilization when both the plant and controller have a
block-lower-triangular structure. Specifically, we give necessary and
sufficient conditions under which a block-lower-triangular stabilizing
controller exists, and we provide a parameterization of all such
controllers akin to the Youla parameterization \cite{youla}.
Many of the results in this section appeared in~\cite{realizability} and similar results also appeared in~\cite{voulgaris_stabilization}.
Throughout this section, we assume 
that the plant $\mathcal{P}$ satisfies
\begin{equation}
  \label{a:minreal}
  \bmat{ \mathcal{P}_{11} & \mathcal{P}_{12} \\
    \mathcal{P}_{21} & \mathcal{P}_{22} } =
  \left[\begin{array}{c|cc}
      A & B_1 & B_2 \\ \hlinet
      C_1 & 0 & D_{12} \\
      C_2 & D_{21} & 0
    \end{array}\right]
\if\MODE2
 \begin{array}{l}\text{is a minimal}\\\text{realization}\end{array}
\else
  \qquad \text{is a minimal realization}
\fi
\end{equation}
We further assume that the
$\mathcal{P}_{22}$ subsystem is block-lower-triangular.  Since we have
$\mathcal{P}_{22} \in \Lower(\Rp)$,
Lemma~\ref{lem:triangular_realization} allows us to assume without loss of generality that the matrices $A$, $B_2$, and $C_2$ have the form
\if\MODE2
	\begin{equation}\label{a:tri_form}
	\addtolength{\arraycolsep}{-0.8mm}
	A \defeq\! \bmat{A_{11} & 0 \\ A_{21} & A_{22}} \,\,
	B_2 \defeq\! \bmat{B_{11} & 0 \\ B_{21} & B_{22}} \,\,
	C_2 \defeq\! \bmat{C_{11} & 0 \\ C_{21} & C_{22}}
	\end{equation}
\else
	\begin{equation}\label{a:tri_form}
	A \defeq \bmat{A_{11} & 0 \\ A_{21} & A_{22}} \qquad
	B_2 \defeq \bmat{B_{11} & 0 \\ B_{21} & B_{22}} \qquad
	C_2 \defeq \bmat{C_{11} & 0 \\ C_{21} & C_{22}}
	\end{equation}
\fi
where $A_{ij} \in\R^{n_i\times n_j}$, $B_{ij}\in\R^{n_i\times
  m_j}$ and $C_{ij}\in\R^{k_i\times n_j}$. The following result gives
necessary and sufficient conditions under which a there exists a
structured stabilizing controller.

\begin{lem}
\label{lem:conditions_stabilizability}
Suppose $\mathcal{P}\in \Rp$ and $\mathcal{P}_{22} \in
\Lower(\Rp,k,m)$. Let $(A,B,C,D)$ be a minimal realization of
$\mathcal{P}$ that satisfies~\eqref{a:minreal}--\eqref{a:tri_form}. There
exists $\mathcal{K}_0\in\Lower(\Rp,m,k)$ such that $\mathcal{K}_0$
stabilizes $\mathcal{P}$ if and only if both
\begin{enumerate}[(i)]
\item $(C_{11},A_{11},B_{11})$ is stabilizable and detectable, and 
\item $(C_{22},A_{22},B_{22})$ is stabilizable and detectable.
\end{enumerate}
In this case, one such controller is
\begin{gather}
\label{eqn:nominal_K0}
  \begin{aligned}
    \mathcal{K}_0 &= \stsp{A+B_2K_d+L_dC_2}{-L_d}{K_d}{0}
  \end{aligned} \\ \notag
\text{where}\quad
K_d \defeq \bmat{K_1&\\& K_2} 
\quad\text{and}\quad 
L_d \defeq  \bmat{L_1&\\&L_2}
\end{gather}
and $K_i$ and $L_i$ are chosen such that $A_{ii}+B_{ii}K_i$ and
\if\MODE2\\\fi
$A_{ii}+L_iC_{ii}$ are Hurwitz for $i=1,2$.
\end{lem}

\begin{proof}
  $(\impliedby)$ Suppose that $(i)$ and $(ii)$ hold. Note that
  $A+B_2K_d$ and $A+L_dC_2$ are Hurwitz by construction, thus
  $(C_2,A,B_2)$ is stabilizable and detectable, and it is thus
  immediate that~\eqref{eqn:nominal_K0} is stabilizing.  Due to the
  block-diagonal structure of $K_d$ and $L_d$, it is straightforward
  to verify that $\mathcal{K}_0 \in \Lower(\Rp)$.

  \noindent
    $(\implies)$ Suppose $\mathcal{K}_0\in\Lower(\Rp,m,k)$
    stabilizes $\mathcal{P}$. Because
  $\mathcal{P}$ is minimal, we must have that the realization
  $(C_2,A,B_2)$ is stabilizable and detectable. By
  Lemma~\ref{lem:triangular_realization}, we may assume that
  $\mathcal{K}_0$ has a minimal realization with $A_K,B_K,C_K,D_K
  \in\Lower(\R)$.  The closed-loop generator $\bar A$ is Hurwitz,
  where
  \[
  \bar A =
  \addtolength{\arraycolsep}{-0.5mm}
  \bmat{A & 0 \\ 0 & A_K} + \bmat{B_2 & 0 \\ 0 & B_K}
  \bmat{I & -D_K \\ 0 & I}^{-1}
  \bmat{0 & C_K \\ C_2 & 0}
  \]
It follows that
\[
\left(\bmat{A & 0 \\ 0 & A_K}, \bmat{B_2 & 0 \\ 0 & B_K}\right) \text{ is stabilizable}
\]
and hence by the PBH test, $(A,B_2)$ and $(A_K,B_K)$ are
stabilizable and similarly $(C_2,A)$ and $(C_K,A_K)$ are
detectable.

Each block of $\bar A$ is block-lower-triangular. Viewing $\bar A$ as a block $4\times 4$ matrix, transform $\bar A$ using a matrix $T$ that permutes states $2$ and $3$.  Now $T^{-1}\bar AT$ is Hurwitz implies that the two $2\times 2$ blocks on the diagonal are Hurwitz. But these diagonal blocks are precisely the closed-loop $\bar A$ matrices corresponding to the 11 and 22
subsystems.
Applying the PBH argument to the diagonal blocks of $T^{-1}\bar AT$
implies that $(C_{ii},A_{ii},B_{ii})$ is stabilizable and detectable
for $i=1,2$ as desired.
\end{proof}

Note that the centralized characterization of stabilization, as in
\cite[Lemma~12.1]{zdg}, only requires that $(C_2,A,B_2)$ be
stabilizable and detectable. The conditions in  Lemma~\ref{lem:conditions_stabilizability} are stronger
because of the additional structural constraint.

We may also characterize the stabilizability of just the 22 block, which we state as a corollary.

\begin{cor}
\label{cor:tristab}
   Suppose $\mathcal{P}_{22} \in \Lower(\Rp,k,m)$, and let
   $(A,B_2,C_2,D_{22})$ be a minimal realization of $\mathcal{P}_{22}$
   that satisfies~\eqref{a:tri_form}.
There exists
   $\mathcal{K}_0\in\Lower(\Rp,m,k)$ such that $\mathcal{K}_0$
   stabilizes~$\mathcal{P}_{22}$ if and only if both
    \begin{enumerate}[(i)]
    \item $(C_{11},A_{11},B_{11})$ is stabilizable and detectable, and 
    \item $(C_{22},A_{22},B_{22})$ is stabilizable and detectable.
    \end{enumerate}
\end{cor}
Indeed, there exist block-lower transfer matrices that cannot be stabilized by a block-lower controller. For example,
\[
\mathcal{P}_{22} = \bmat{\frac{1}{s+1} & 0 \\ \T \frac{1}{s-1} &
  \frac{1}{s+1}} = \left[\begin{array}{ccc|cc}
    -1 & 0 & 0 & 1 & 0 \\
    0 & 1 & 0 & 1 & 0 \\
    0 & 0 & -1 & 0 & 1 \\ \hlinet
    1 & 0 & 0 & 0 & 0 \\
    0 & 1 & 1 & 0 & 0
\end{array}\right]
\]
The above realization is minimal, but the grouping of states into
blocks is not unique.  We may group the unstable mode either
in the $A_{11}$ block or in the $A_{22}$ block, which corresponds to
$n=(2,1)$ or $n=(1,2)$, respectively.  The former leads to an undetectable $(C_{11},A_{11})$ while the latter leads to an unstabilizable $(A_{22},B_{22})$.
By Corollary~\ref{cor:tristab}, this plant cannot be
stabilized by a block-lower-triangular controller. However,
centralized stabilizing controllers exist due to the minimality of the
realization.

Note that a stabilizable $\mathcal{P}_{22}$ may have an off-diagonal block that is unstable. An example is:
\[
\mathcal{P}_{22} = \bmat{\frac{1}{s-1} & 0 \\ \T \frac{1}{s-1} &
  \frac{1}{s+1}} \qquad \mathcal{K}_0 = \bmat{ -2 & 0 \\
  \T 0 & 0\vphantom{\frac{1}{s-1}} }
\]
We now treat the parameterization of all stabilizing controllers.
The following result was proved in~\cite{voulgaris_stabilization}.

\begin{thm}\label{thm:stabilizing_controllers}  
Suppose the conditions of
Lemma~\ref{lem:conditions_stabilizability} hold, and 
$(C_{11},A_{11},B_{11})$ and $(C_{22},A_{22},B_{22})$ are both
stabilizable and detectable. Define $K_d$ and $L_d$ as in 
Lemma~\ref{lem:conditions_stabilizability}.
The set of all
$\mathcal{K}\in\Lower(\Rp,m,k)$ that stabilize $\mathcal{P}$ is
 \begin{equation*}
 \bl\{\mathcal{F}_\ell(\mathcal{J}_d, \mathcal{Q})
 \,\vertt\,
 \mathcal{Q} \in  \Lower(\RHinf,m,k) \br\}
 \end{equation*}
\begin{equation}\label{eqn:J}
\text{where}\quad\mathcal{J}_d = 
  \left[\begin{array}{c|cc}
      A+B_2K_d+L_dC_2
      & -L_d & B_2
      \\ \hlinet
      K_d & 0 & I \\
      -C_2 & I & 0
    \end{array}\right]
\end{equation}
\end{thm}

\begin{proof}
  If we relax the constraint that $\mathcal{Q}$ be lower triangular,
  then this is the standard parameterization of all centralized
  stabilizing controllers~\cite[Theorem 12.8]{zdg}.  It suffices to
  show that the map from $\mathcal{Q}$ to $\mathcal{K}$ and its
  inverse are structure-preserving.  Since each block of the
  state-space realization of $\mathcal{J}_d$ is in $\Lower(\R)$, we
  have $(\mathcal{J}_d)_{ij}\in\Lower(\Rp)$. Thus~$F_\ell(\mathcal{J}_d,
  \cdot)$ preserves lower triangularity on its domain.  This also
  holds for the inverse map $F_u(\mathcal{J}^{-1}_d, \cdot)$, since
  \[
  \mathcal{J}_d^{-1} = \left[\begin{array}{c|cc}
      A & B_2 & -L_d \\ \hlinet
      C_2 & 0 & I \\
      -K_d & I & 0
    \end{array}\right]
  \]\vskip-\baselineskip
\end{proof}

As in the centralized case, this parameterization of stabilizing
controllers allows us to rewrite the closed-loop map in terms of
$\mathcal{Q}$. After some simplification, we obtain
$
  F_\ell(\mathcal{P},\mathcal{K}) =  \mathcal{T}_{11} + 
  \mathcal{T}_{12}  \mathcal{Q} \mathcal{T}_{21}
$
  where $\mathcal{T}$ has the realization  
  \begin{equation}
    \label{eqn:T}
    \bmat{\mathcal{T}_{11} & \mathcal{T}_{12} \\ \mathcal{T}_{21} & 0} =
    \left[\begin{array}{cc|cc}
        A_{Kd} &  -B_2 K_d & B_1 & B_2 \\
        0 & A_{Ld} & B_{Ld} & 0 \\ \hlinet
        C_{Kd} & -D_{12}K_d & 0 & D_{12} \\
        0 & C_2 & D_{21} & 0
      \end{array}\right]
  \end{equation}
and we have used the following shorthand notation.
\begin{equation}\label{defn:short}
\begin{aligned}
A_{Kd} &\defeq A+B_2K_d 		& A_{Ld} &\defeq A+L_dC_2 \\
C_{Kd} &\defeq C_1+D_{12}K_d 	& B_{Ld} &\defeq B_1+L_dD_{21}
\end{aligned}
\end{equation} 
Combining the results above gives the following important equivalence between the two-player output-feedback problem~\eqref{opt:2p_output_feedback} and a structured model-matching problem.
  \begin{cor}
    \label{cor:optimization_model_match}
    Suppose the conditions of
    Theorem~\ref{thm:stabilizing_controllers} hold.  Then
    $\mathcal{Q}_\textup{opt}$ is a minimizer~for
   \begin{equation}
      \label{opt:2pmm_gen}
      \begin{aligned}
        \textup{minimize}\qquad& \bigl\| \mathcal{T}_{11} +
        \mathcal{T}_{12} \mathcal{Q} \mathcal{T}_{21} \bigr\|_2 \\
        \textup{subject to}\qquad& \mathcal{Q} \in \Lower(\RHinf)
      \end{aligned}
    \end{equation}
   if and only if $\mathcal{K}_\textup{opt}=
   F_\ell(\mathcal{J}_d,\mathcal{Q}_\textup{opt})$ is a minimizer for
   the the two-player output-feedback problem
   \eqref{opt:2p_output_feedback}.
Here $\mathcal{J}_d$
   is given by \eqref{eqn:J}, and $\mathcal{T}$ is defined in~\eqref{eqn:T}.
    Furthermore, $\mathcal{Q}_\textup{opt}$ is unique if and only if
    $\mathcal{K}_\textup{opt}$ is unique.
  \end{cor}

Corollary~\ref{cor:optimization_model_match} gives an equivalence between the two-player output-feedback problem \eqref{opt:2p_output_feedback} and the 
two-player stable model-matching problem \eqref{opt:2pmm_gen}. The new
formulation should be easier to solve than the output-feedback version
because it is convex and the~$\mathcal{T}_{ij}$ are stable. However,
its solution is still not straightforward, because the problem remains
infinite-dimensional and there is a structural constraint on
$\mathcal{Q}$.

\section{Main result}

\label{sec:main}

In this section, we present our main result: a solution to the
two-player output-feedback problem. First, we state our assumptions
and list the equations that must be solved. We assume the plant satisfies~\eqref{a:minreal} and the matrices $A$, $B_2$, $C_2$ have the form~\eqref{a:tri_form}. We further assume that $A_{11}$ and $A_{22}$ have non-empty dimensions, so $n_i\ne 0$. This avoids trivial special cases and allows us to streamline the results. To ease notation, we define the following cost and
covariance matrices
\begin{equation}
  \label{a:QRWV}
  \begin{aligned}
    \bmat{Q & S \\ S^\tp & R } &\defeq
    \if\MODE2\addtolength{\arraycolsep}{-0.5mm}\fi
    \bmat{ C_1^\tp C_1 & C_1^\tp D_{12} \\[.5mm] D_{12}^\tp C_1 & D_{12}^\tp D_{12} } =
    \if\MODE2\addtolength{\arraycolsep}{-0.7mm}\fi
    \bmat{C_1 & D_{12}}^\tp \!\bmat{C_1 & D_{12}} \\
    \if\MODE2\addtolength{\arraycolsep}{-0.5mm}\fi
    \bmat{W & U^\tp \\ U & V} &\defeq
    \if\MODE2\addtolength{\arraycolsep}{-0.7mm}\fi
    \bmat{ B_1B_1^\tp & B_1 D_{21}^\tp \\[.5mm] D_{21}B_1^\tp & D_{21} D_{21}^\tp} =
    \if\MODE2\addtolength{\arraycolsep}{-0.5mm}\fi
    \bmat{B_1 \\ D_{21}}\bmat{B_1 \\ D_{21}}^\tp
  \end{aligned}\if\MODE2\negskip\fi
\end{equation}
Our main assumptions are as follows.
\begin{enumerate}[\it{A}1)]
\if\MODE2\itemsep=1.5mm\fi
\item $D_{12}^\tp D_{12} > 0$ \label{ass:Afirst}
\item $(A_{11},B_{11})$ and $(A_{22},B_{22})$ are stabilizable
\item $\bmat{A -\i\omega I & B_2 \\ C_1 & D_{12}}$ has full column rank for all $\omega \in \R$
\item $ D_{21}D_{21}^\tp > 0$
\item $(C_{11},A_{11})$ and $(C_{22},A_{22})$ are detectable
\item $\bmat{A -\i\omega I & B_1 \\ C_2 & D_{21}}$ has full row rank for all $\omega \in \R$
\label{ass:Alast}
\end{enumerate}
We will also require the solutions to four AREs
\begin{equation}\label{tp:ares}
\hspace{-2mm}
\begin{aligned} 
(X,K) &= \are(A,B_2,C_1,D_{12}) \if\MODE2\\\else&\fi
(Y,L^\tp) &= \are(A^\tp,C_2^\tp,B_1^\tp,D_{21}^\tp)   \\
(\tilde X,J) &= \are(A_{22},B_{22},C_1E_2,D_{12}E_2) \if\MODE2\\\else&\fi
(\tilde Y,M^\tp) &= \are(A_{11}^\tp,C_{11}^\tp,E_1^\tp B_1^\tp,E_1^\tp D_{21}^\tp) 
\end{aligned}
\end{equation}
Finally, we must solve the following simultaneous linear equations for
$\Phi,\Psi\in\R^{n_2 \times n_1}$
\begin{multline} \label{tp:phi}
(A_{22}+B_{22}J)^\tp \Phi + \Phi (A_{11}+ MC_{11}) \if\MODE2\\\fi
-(\tilde X - X_{22})(\Psi C_{11}^\tp + U_{12}^\tp)V_{11}^{-1}C_{11}\\
+ \bigl(  \tilde X A_{21} + J^\tp S_{12}^\tp + Q_{21} - X_{21}MC_{11} \bigr) = 0
\end{multline}
\if\MODE2\vspace{-8mm}\fi
\begin{multline} \label{tp:psi}
(A_{22}+B_{22}J) \Psi + \Psi (A_{11}+ MC_{11})^\tp  \if\MODE2\\\fi
-B_{22}R_{22}^{-1}(B_{22}^\tp \Phi + S_{12}^\tp)(\tilde Y - Y_{11})\\
+ \bigl(  A_{21}\tilde Y + U_{12}^\tp M^\tp + W_{21} - B_{22}JY_{21} \bigr)= 0
\end{multline}
and define the associated gains $\hat K \in\R^{(m_1+m_2) \times (n_1+n_2)}$ and $\hat L\in\R^{(n_1 + n_2)\times (k_1+k_2)}$
\begin{align}\label{tp:gains}
\hat K &\defeq \bmat{ 0 & 0 \\-R_{22}^{-1} \left( B_{22}^\tp \Phi + S_{12}^\tp \right) & J}
\if\MODE2\\\else&\fi
\hat L &\defeq \bmat{ M &0\\ -(\Psi C_{11}^\tp+U_{12}^\tp )V_{11}^{-1}& 0}
\end{align}
For convenience, we define the Hurwitz matrices
\begin{equation}\label{a:ahat}
\begin{aligned}
A_K &\defeq A+B_2 K & 			A_L &\defeq A+LC_2 \\
A_J &\defeq A_{22}+B_{22}J &		A_M &\defeq A_{11}+MC_{11} \\ 
\hat A &\defeq A+B_2\hat K+\hat L C_2
\end{aligned}
\end{equation}
Note that $A_K,A_L,A_J,A_M$ are all Hurwitz by construction, and $\hat
A$ is Hurwitz as well, because it is block-lower-triangular and its
block-diagonal entries are $A_M$ and $A_J$.  The matrices $\Phi$ and
$\Psi$ have physical interpretations, as do the gains $\hat K$ and
$\hat L$. These will be explained in Section~\ref{sec:struct}. The
main result of this paper is Theorem~\ref{thm:main}, given below.

\begin{thm}\label{thm:main} 
Suppose $\mathcal{P}\in\Rp$
satisfies~\eqref{a:minreal}--\eqref{a:tri_form} and Assumptions
A\ref{ass:Afirst}--A\ref{ass:Alast}. Consider the two-player output
feedback problem as stated in~\eqref{opt:2p_output_feedback}
  \[
  \begin{aligned}
    \textup{minimize}\qquad
    & \big\|
      \mathcal{P}_{11} + \mathcal{P}_{12} 
      \mathcal{K}(I-\mathcal{P}_{22}\mathcal{K})^{-1} \mathcal{P}_{21} 
     \big\|_2 \\
    \textup{subject to}\qquad& \mathcal{K} \in \Lower(\Rp,m,k) \\
    & \mathcal{K}\textup{ stabilizes }\mathcal{P}
  \end{aligned}
  \]
\begin{enumerate}[(i)]
\if\MODE2\itemsep=1mm\fi
\item There exists a unique optimal $\mathcal{K}$.

\item Suppose $\Phi,\Psi$ satisfy the linear
  equations~\eqref{tp:phi}--\eqref{tp:psi}. Then the optimal
  controller is
  \begin{equation}\if\MODE2\hspace{-7mm}\fi
    \label{eqn:kopt}
    \if\MODE2\addtolength{\arraycolsep}{-0.5mm}\fi
    \mathcal{K}_\textup{opt} = \left[\begin{array}{cc|c}
        A+B_2K+\hat L C_2 & 0 & -\hat L  \\
        B_2K-B_2\hat K& A+LC_2+B_2\hat K & -L \\ \hlinet
        K - \hat K & \hat K & 0 
      \end{array}\right]
  \end{equation}
where $K$, $L$, $\hat K$, $\hat L$ are defined in~\eqref{tp:ares} and~\eqref{tp:gains}.
  
\item There exist  $\Phi,\Psi$ satisfying~\eqref{tp:phi}--\eqref{tp:psi}.
\end{enumerate}
\end{thm}
\begin{proof}
A complete proof is provided in Section~\ref{sec:main_proof}.
\end{proof}

\noindent 
An alternative realization for the optimal controller is 
\begin{equation}\label{eqn:kopt2}
\if\MODE2\addtolength{\arraycolsep}{-1.1mm}\fi
  \mathcal{K}_\textup{opt} = \left[\begin{array}{cc|c}
      A+B_2K+\hat L C_2 & 0 & \hat L  \\
      LC_2-\hat L C_2 & A+LC_2+B_2\hat K & L-\hat L \\ \hlinet
      -K & -\hat K & 0 
    \end{array}\right]
\end{equation}


\section{Structure of the optimal controller}
\label{sec:struct}

In this section, we will discuss several features and consequences of the optimal controller exhibited in Theorem~\ref{thm:main}.
First, a brief discussion on duality and symmetry.  We
then show a structural result, that the states of the optimal
controller have a natural stochastic interpretation as minimum-mean-square-error estimates. This will lead to a useful interpretation of the matrices $\Phi$ and $\Psi$ defined in \eqref{tp:phi}--\eqref{tp:psi}.  We then compute the $\Htwo$-norm of the optimally controlled system, and characterize the portion of the cost attributed to the decentralization constraint.  Finally, we show how our main result specializes to many previously solved cases appearing in the literature.

\subsection{Symmetry and duality}

The solution to the two-player output-feedback problem has
nice symmetry properties that are perhaps unexpected given the network
topology. Player~2 sees more information than Player~1, so one
might expect Player~2's optimal policy to be more complicated
than that of Player~1. Yet, this is not the case.  Player~2 observes
all the measurements but only controls $u_2$, so only influences
subsystem 1.  In contrast, Player~1 only observes subsystem~1, but 
controls $u_1$, which in turn influences both subsystems. This duality is
reflected in
\eqref{tp:ares}--\eqref{tp:gains}.

If we transpose the system
variables and swap Player~1 and Player~2, then every quantity related
to control of the original system becomes a corresponding quantity
related to estimation of the transformed system. More formally, if we
define
\[
A^\ddag = \bmat{A_{22}^\tp & A_{12}^\tp \\[.5mm] A_{21}^\tp & A_{11}^\tp}
\]
then the transformation $(A,B_2,C_2) \mapsto
(A^\ddag,C_2^\ddag,B_2^\ddag)$ leads to
\begin{align*}
  (X,K) \mapsto (Y^\ddag,L^\ddag), \,\,\,
  (\hat X,\hat K) \mapsto (\hat Y^\ddag,\hat L^\ddag), \,\,\,\text{and}\,\,\,
  \hat A \mapsto \hat A^\ddag.
\end{align*}
This is analogous to the well-known correspondence between Kalman
filtering and optimal control in the centralized case.


\subsection{Gramian equations}

In this subsection, we derive useful algebraic relations that follow
from Theorem~\ref{thm:main} and provide a stochastic
interpretation. Recall the plant equations, written in differential
form.
\begin{align}\label{eq:ssplant}
\begin{aligned}
\dot x &= Ax + B_1 w + B_2 u \\
y &= C_2 x + D_{21} w
\end{aligned}
\end{align}
The optimal controller~\eqref{eqn:kopt} is given in
Theorem~\ref{thm:main}. Label its states as $\zeta$ and $\xi$, and
place each equation into the following canonical observer form.
\begin{align}
  \dot\zeta &= A\zeta +B_2\hat u - \hat L (y-C_2\zeta) \label{e1} \\
  \hat u &= K\zeta \label {e1b} \\
  \dot\xi &= A\xi + B_2 u - L (y-C_2\xi) \label{e2} \\
  u &= K\zeta + \hat K (\xi - \zeta) \label{e3}
\end{align}
We will see later that this choice of coordinates leads to a natural
interpretation in terms of steady-state Kalman filters. Our first
result is a Lyapunov equation unique to the two-player problem.

\begin{lem}\label{lem:lyapY}
Suppose $\Phi,\Psi$ is a solution of~\eqref{tp:phi}--\eqref{tp:psi}. 
Define $\hat L$ and $\hat A$ according to~\eqref{tp:gains}--\eqref{a:ahat}.
There exists a unique matrix~$\hat Y$ satisfying the equation
\begin{equation}\label{ee:lyapY}
\hat A (\hat Y - Y) + (\hat Y - Y)\hat A^\tp + (\hat L - L)V(\hat L - L)^\tp = 0
\end{equation}
Further, $\hat Y \geq Y$,  $\hat Y_{11} = \tilde Y$, $\hat Y_{21} = \Psi$,    and
\begin{align}
\label{ef1}
\hat L &=
 -\bl( \hat YC_2^\tp + U^\tp \br) E_1 V_{11}^{-1}E_1^\tp  
\end{align}
\end{lem}

\begin{proof}
  Since $\hat A$ is stable and $(\hat L - L)V(\hat L-L)^\tp \geq 0$,
  it follows from standard properties of~Lyapunov equations that
  \eqref{ee:lyapY} has a unique solution and it satisfies $\hat Y -Y \geq 0$.
  Right-multiplying~\eqref{ee:lyapY} by~$E_1$ gives
  \if\MODE2
	\begin{multline*}
	\hat A (\hat Y E_1 - YE_1) + (\hat YE_1 - YE_1)A_M^\tp \\
	+ (\hat L - L)V(E_1M^\tp - L^\tp E_1) = 0
	\end{multline*}
  \else
	\[
	\hat A (\hat Y E_1 - YE_1) + (\hat YE_1 - YE_1)A_M^\tp +
	(\hat L - L)V(E_1M^\tp - L^\tp E_1) = 0
	\]
  \fi
  This equation splits into two Sylvester equations, the first in
  $\hat{Y}_{11}$, and the second in $\hat{Y}_{21}$.  Both 
  have unique solutions.  Upon comparison
  with \eqref{tp:ares}, \eqref{tp:psi}, \eqref{tp:gains}, one
  finds that $\hat Y_{11} = \tilde Y$ and $\hat Y_{21}
  = \Psi$. Note that, since $\Phi$ is fixed, $\Psi$ is uniquely determined. 
Similarly comparing with~\eqref{tp:gains}
  verifies \eqref{ef1}.
\end{proof}

Our next observation is that under the right coordinates, the
controllability~Gramian of the closed-loop map is block-diagonal.
\begin{thm}\label{thm:diag_gram}
Suppose we have the plant described by~\eqref{eq:ssplant} and the
optimal controller given by~\eqref{e1}--\eqref{e3}. The closed-loop
map for one particular choice of coordinates is
\if\MODE2
	\begin{multline*}
	\bmat{\dot\zeta \\ \dot{\xi}-\dot{\zeta} \\ \dot x - \dot \xi} =
	\bmat{A_K & -\hat L C_2 & -\hat L C_2 \\ 0 & \hat A & (\hat L - L)C_2 \\ 0 & 0 & A_L}
	\!\bmat{\zeta \\ \xi-\zeta \\ x-\xi} \\
	+ \bmat{-\hat L D_{21}\\ (\hat L - L)D_{21} \\  B_1+LD_{21} } w
	\end{multline*}
\else
	\[
	\bmat{\dot\zeta \\ \dot{\xi}-\dot{\zeta} \\ \dot x - \dot \xi} =
	\bmat{A_K & -\hat L C_2 & -\hat L C_2 \\ 0 & \hat A & (\hat L - L)C_2 \\ 0 & 0 & A_L}
	\bmat{\zeta \\ \xi-\zeta \\ x-\xi} +
	\bmat{-\hat L D_{21}\\ (\hat L - L)D_{21} \\  B_1+LD_{21} } w
	\]
\fi
which we write compactly as $\dot q = A_c q + B_c w$. Let $\Theta$ be
the controllability~Gramian, i.e. the solution $\Theta \ge 0$ to the
Lyapunov equation $A_c \Theta + \Theta A_c^\tp + B_c B_c^\tp =
0$. Then
\[
\Theta = \bmat{ Z & 0 & 0 \\ 0 & \hat Y - Y & 0 \\ 0 & 0 & Y }
\]
where $Z$ satisfies the Lyapunov equation
\[
A_K Z + ZA_K^\tp + \hat L V \hat L^\tp = 0
\]
\end{thm}

\begin{proof}
Uniqueness of $\Theta$ follows because $A_c$ is Hurwitz, and $\Theta
\ge 0$ follows because $A_c$ is Hurwitz and $B_cB_c^\tp \ge 0$. The
solution can be verified by direct substitution and by making use of
the identities in Lemma~\ref{lem:lyapY}.
\end{proof}

The observation in Theorem~\ref{thm:diag_gram} lends itself to a
statistical interpretation. If $w$ is white Gaussian noise with unit
intensity, the steady-state distribution of the random vector~$q$ will
have covariance~$\Theta$. Since $\Theta$ is block-diagonal, the
steady-state distributions of the components $\zeta$, $\xi-\zeta$, and
$x-\xi$ are mutually independent.  As expected from our discussion on
duality, the algebraic relations derived in Lemma~\ref{lem:lyapY} may
be dualized to obtain a new result, which we now state without proof.

\begin{lem}
  \label{lem:lyapX}
  Suppose $\Phi,\Psi$ is a solution of~\eqref{tp:phi}--\eqref{tp:psi}. Define $\hat K$ and $\hat A$ according to~\eqref{tp:gains}--\eqref{a:ahat}.
  There exists a unique matrix~$\hat X$ satisfying the equation
 \begin{equation}
  \label{ee:lyapX}
  \hat A^\tp (\hat X - X) + (\hat X - X)\hat A + (\hat K - K)^\tp R(\hat K - K) = 0
 \end{equation}
 Further,  $\hat X \geq X$, $\hat X_{22} = \tilde X$, $\hat X_{21} = \Phi$, and
 \begin{align}
   \label{eg1}
   \hat K &=
   -E_2R_{22}^{-1}E_2^\tp \bl( B_2^\tp \hat X+ S^\tp \br)
 \end{align}
\end{lem}

The dual of Theorem~\ref{thm:diag_gram} can be obtained by expressing
the closed-loop map in the coordinates $(x,x-\zeta,x-\xi)$ instead. In
these coordinates, one can show that the observability~Gramian is also
block-diagonal.


\subsection{Estimation structure}

In this subsection, we show that the states $\zeta$ and $\xi$ of
optimal controller~\eqref{e1}--\eqref{e3} may be interpreted as
suitably defined Kalman filters. Given the
dynamics~\eqref{eq:ssplant}, the estimation map and corresponding Kalman filter $\hat x$ is given by
\if\MODE2
\begin{equation}
  \label{eqn:kf1}
  \bmat{ x \\ y} \!=\!
  \addtolength{\arraycolsep}{-2pt}
  \left[\begin{array}{c|cc}
      A & B_1 & B_2 \\ \hlinet
      I & 0 & 0 \\
      C_2 & D_{21} & 0
    \end{array}\right]\!\bmat{w \\ u}\!,\;
  \hat x \!=\!
  \addtolength{\arraycolsep}{-1pt}
  \left[\begin{array}{c|cc}
      A_L & -L & B_2 \\ \hlinet
      I & 0 & 0
    \end{array}\right]\!\bmat{y \\ u}
\end{equation}
\else
\begin{equation}
  \label{eqn:kf1}
  \bmat{ x \\ y} =
  \left[\begin{array}{c|cc}
      A & B_1 & B_2 \\ \hlinet
      I & 0 & 0 \\
      C_2 & D_{21} & 0
    \end{array}\right]\bmat{w \\ u},\qquad
  \hat x = \left[\begin{array}{c|cc}
      A+LC_2 & -L & B_2 \\ \hlinet
      I & 0 & 0
    \end{array}\right] \bmat{y \\ u}
\end{equation}
\fi
where $L$ and $A_L$ are defined in~\eqref{tp:ares} and~\eqref{a:ahat} respectively. We will show that this Kalman
filter is the result of an $\Htwo$ optimization problem and we will
define it as such. This operator-based definition will allow us to
make precise claims without the need for stochastics. Before stating
our definition, we illustrate an important subtlety regarding the
input~$u$; in order to make the notion of a Kalman filter unambiguous,
one must be careful to specify which signals will be treated as
inputs. Consider the following numerical instance of~\eqref{eqn:kf1}.
\if\MODE2
\begin{align}
  \label{eqn:kf2}
  \bmat{x \\ y} \!=\!
  \addtolength{\arraycolsep}{-1pt}
  \left[\begin{array}{c|cc;{3pt/2pt}c}
      -4 & 3 & 0 & 1 \\ \hlinet
      1 & 0 & 0 & 0 \\
      1 & 0 & 1 & 0
    \end{array}\right]\!\bmat{w \\ u}\!,\quad
  \hat x \!=\!
  \addtolength{\arraycolsep}{-1pt}
  \left[\begin{array}{c|cc}
      -5 & 1 & 1 \\ \hlinet
      1 & 0 & 0
    \end{array}\right]\!\bmat{y \\ u}
\end{align}
\else
\begin{equation}
  \label{eqn:kf2}
  \bmat{x \\ y} = \left[\begin{array}{c|cc;{3pt/2pt}c}
      -4 & 3 & 0 & 1 \\ \hlinet
      1 & 0 & 0 & 0 \\
      1 & 0 & 1 & 0
    \end{array}\right]\bmat{w \\ u},\qquad
  \hat x = \left[\begin{array}{c|cc}
      -5 & 1 & 1 \\ \hlinet
      1 & 0 & 0
    \end{array}\right]
  \bmat{y \\ u}
\end{equation}
\fi
Suppose we know the control law $u$. We may then eliminate~$u$ from the estimator~\eqref{eqn:kf2}. For example,
\begin{align}\label{eqn:kf3}
u = \stsp{1}{1}{1}{0}y
\quad\text{leads to} \quad
\hat x = \left[\begin{array}{cc|c}
-5 & 1 & 1 \\
0 & 1 & 1 \\ \hlinet
1 & 0 & 0
\end{array}\right]y
\end{align}
The new estimator dynamics are unstable because we chose an unstable control law. Consider estimating $x$ again, but this time using the control law a~priori. Eliminating $u$ from the plant first and then computing the Kalman filter leads to
\if\MODE2
\begin{align}\label{eqn:kf4}
\bmat{x \\ y} =
\addtolength{\arraycolsep}{-1pt}
\left[\begin{array}{cc|cc}
-4 & 1 & 3 & 0 \\ 1 & 1 & 0 & 1\\ \hlinet
1 & 0 & 0 & 0\\ 1 & 0 & 0 & 1
\end{array}\right]w,\quad
\hat x = 
\addtolength{\arraycolsep}{-1pt}
\left[\begin{array}{cc|c}
-7 & 1 & 3\\
-12 & 1 & 13\\ \hlinet
1 & 0 & 0 
\end{array}\right]y
\end{align}
\else
\begin{align}\label{eqn:kf4}
\bmat{x \\ y} =
\left[\begin{array}{cc|cc}
-4 & 1 & 3 & 0 \\ 1 & 1 & 0 & 1\\ \hlinet
1 & 0 & 0 & 0\\ 1 & 0 & 0 & 1
\end{array}\right]w,\quad
\hat x = 
\left[\begin{array}{cc|c}
-7 & 1 & 3\\
-12 & 1 & 13\\ \hlinet
1 & 0 & 0 
\end{array}\right]y
\end{align}
\fi
The estimators~\eqref{eqn:kf3} and~\eqref{eqn:kf4} are different. Indeed,
\[
\hat x = \left(\frac{s}{s^2+4s-5}\right)y
\quad\text{and}\quad
\hat x = \left(\frac{3s+10}{s^2+6s+5}\right)y
\]
in~\eqref{eqn:kf3} and~\eqref{eqn:kf4} respectively. In general, open-loop and closed-loop estimation are different so any sensible definition of Kalman filtering must account for this by specifying which inputs (if any) will be estimated in open-loop. For centralized optimal control, the correct filter to use is given by~\eqref{eqn:kf1}--\eqref{eqn:kf2}.

Our definition of an \emph{estimator} is as follows. Note that similar optimization-based definitions have appeared in the literature, as in~\cite{sss}. 

\begin{defn}\label{def:kfc}
  Suppose $\mathcal{G}\in\Rp^{(p_1+p_2)\times(q_1 +q_2)  }$ is given by 
  \[
  \mathcal{G} = \bmat{
    \mathcal{G}_{11} & \mathcal{G}_{12} \\
    \mathcal{G}_{21} & \mathcal{G}_{22}    
  }
  \]
  and $\mathcal{G}_{21}(\i\omega)$ has full row rank for all
  $\omega\in\R\cup\{\infty\}$. Define the estimator of
  $\mathcal{G}$ to be $\mathcal{G}^\textup{est}\in\Rp^{p_1 \times (p_2
    + q_1)}$
  partitioned according to $  \mathcal{G}^\textup{est}
  = 
  \bmat{\mathcal{G}^\textup{est}_1 & \mathcal{G}^\textup{est}_2}
    $
  where
  \begin{align*}
  \mathcal{G}^\textup{est}_1 
  & = \argmin_{\substack{\mathcal{F}\in\RHtwo\\[1pt]\mathcal{G}_{11} -
      \mathcal{F} \mathcal{G}_{21}\in\RHtwo}} 
  \normm{\mathcal{G}_{11} - \mathcal{F} \mathcal{G}_{21}}_{2}\\
  \mathcal{G}^\textup{est}_2
  &= \mathcal{G}_{12} -   \mathcal{G}^\textup{est}_1 \mathcal{G}_{22}
  \end{align*}
\end{defn}
Note that under the assumptions of Definition~\ref{def:kfc}, the
estimator $\mathcal{G}^\textup{est}$ is unique. We will show existence
of  $\mathcal{G}^\textup{est}$ 
for particular $\mathcal{G}$ below. We now define the
following notation.
\begin{defn}\label{def:must}
  Suppose $\mathcal{G}$ and $\mathcal{G}^\textup{est}$ are
  as in Definition~\ref{def:kfc}. If
  \[
  \bmat{x \\ y} = \mathcal{G}\bmat{w \\ u} 
  \]
  then we use the notation $x_{|y,u}$ to mean
  \[
  x_{|y,u} \defeq \mathcal{G}^\textup{est}\bmat{y \\ u}
  \]
\end{defn}
\noindent 
This notation is motivated by the property that
\[
x-  x_{|y,u}
  =( \mathcal{G}_{11}  -\mathcal{G}^\textup{est}_1 \mathcal{G}_{21})w
\]
The stochastic interpretation of Definition~\ref{def:kfc} is therefore that $\mathcal{G}^\textup{est}_1 $ is chosen to minimize the mean square estimation error assuming $w$ is white Gaussian noise with unit intensity. In the following lemma, we show that the quantity $x_{|y,u}$ as defined in Definitions~\ref{def:kfc} and~\ref{def:must} is the usual steady-state Kalman filter for estimating the state~$x$ using the measurements $y$ and inputs $u$, as given in~\eqref{eqn:kf1}.

\begin{lem}\label{lem:structure_i}
  \label{lem:structure_centralized}
  Let $\mathcal{G}$ be 
  \[
  \mathcal{G} = \left[\begin{array}{c|cc}A & B_1 & B_2\\\hlinet 
      I & 0 & 0 \\ C_2 & D_{21} & 0 
    \end{array}\right]
  \]
  and suppose Assumptions~A4--A6 hold. Then
    \[
    \mathcal{G}^\textup{est} = 
    \left[\begin{array}{c|cc}
        A_L & -L  & B_2\\\hlinet 
        I & 0 & 0
      \end{array}\right]
    \]
    where $L$ and $A_L$ are given by~\eqref{tp:ares} and~\eqref{a:ahat} respectively.
  \end{lem}
  \begin{proof}
  A proof of a more general result that includes preview appears in~\cite[Theorem~4.8]{kr_thesis}. Roughly, one begins by applying the change of variables $\mathcal{F} = \mathcal{G}^\textup{est}_1 + \bar{\mathcal{F}}$, and then parameterizing all admissible $\bar{\mathcal{F}}$ using a coprime factorization of $\mathcal{G}_{21}$. This yields the following equivalent unconstrained problem
\if\MODE2
    \begin{align*}
      \textup{min} & \quad 
      \bbbbl\|
        \stsp{A_L}{B_1+LD_{21}}{I}{0} \bbbbr.
        - \bbbbl.\mathcal{Q}\stsp{A_L}{B_1+LD_{21}}{C_2}{D_{21}} \bbbbr\|_2
      \\
      \textup{s.t.} & \quad \mathcal{Q}\in\RHtwo
    \end{align*}
\else
    \begin{align*}
      \textup{minimize} & \quad 
      \normmmmm{
        \stsp{A_L}{B_1+LD_{21}}{I}{0} -
        \mathcal{Q}\stsp{A_L}{B_1+LD_{21}}{C_2}{D_{21}} }_2
      \\
      \textup{subject to} & \quad \mathcal{Q}\in\RHtwo
    \end{align*}
\fi
One may check that $\mathcal{Q}=0$ is optimal for the unconstrained problem, and hence $\bar{\mathcal{F}}=0$ as required, by verifying the following orthogonality condition.
\begin{equation}\label{orth}
  \if\MODE2\addtolength{\arraycolsep}{-1mm}\fi
  \stsp{A_L}{B_1+LD_{21}}{I}{0}\stsp{A_L}{B_1+LD_{21}}{C_2}{D_{21}}^*
  \in \Htwo^\perp
\end{equation}
To see why \eqref{orth} holds, multiply the realizations together and use the state transformation $(x_1,x_2)\mapsto (x_1-Yx_2,x_2)$, where $Y$ is given in~\eqref{tp:ares}. See~\cite[Lemma~14.3]{zdg} for an example of such a computation.
  \end{proof}

Comparing the result of Lemma~\ref{lem:structure_i} with the
$\xi$-state of the optimal two-player controller~\eqref{e2}, we notice
the following consequence.

\begin{cor} Suppose the conditions of Theorem~\ref{thm:main} are satisfied,
and the controller states are labeled as in~\eqref{e1}--\eqref{e3}.
Then
$
x_{|y,u} = \xi
$,
where the estimator is defined by the map
$\mathcal{G}:(w,u)\to(x,y)$ induced by the plant~\eqref{eq:ssplant}.
\end{cor}

The above result means that $\xi$, one of the states of the optimal
controller, is the usual Kalman filter estimate of $x$ given $y$ and
$u$.  The next result is more difficult, and gives the analogous
result for the other state.  The next theorem is our main structural
result. We show that $\zeta$ may also be interpreted as an optimal
estimator in the sense of Definitions~\ref{def:kfc} and~\ref{def:must}.

\begin{thm}\label{thm:structure_tpof}
Suppose the conditions of Theorem~\ref{thm:main} are satisfied,
and label the states of the controller as in~\eqref{e1}--\eqref{e3}.
Define also 
\begin{equation}
  \label{eqn:lotsofu}
  \begin{aligned}
    u_\zeta \defeq (K-\hat K)\zeta,\qquad
    u_\xi  \defeq \hat K \xi
  \end{aligned}
\end{equation}
so that $u = u_\zeta + u_\xi$. Then
\[
\bmat{ x \\ \xi \\ u}_{|y_1, u_\zeta} = \bmat{ \zeta \\ \zeta \\ \hat u }
\]
The estimator is defined by the map  
$\mathcal{G}:(w,u_\zeta)\to(x,\xi,u,y_1)$ induced by~\eqref{eqn:lotsofu}, the
plant~\eqref{eq:ssplant}, and the controller~\eqref{e1}--\eqref{e3}.
\end{thm}

\begin{proof}
The proof parallels that of Lemma~\ref{lem:structure_i}, so we omit
most of the details. Straightforward algebra gives
\[
\mathcal{G} = \left[\begin{array}{cc|cc}
	A & B_2\hat K & B_1 & B_2 \\
	-LC_2 & A+LC_2+B_2\hat K & -LD_{21} & B_2 \\ \hlinet
	I & 0 & 0 & 0\\
	0 & I & 0 & 0\\
	0 & \hat K & 0 & I \\
	E_1^\tp C_2 & 0 & E_1^\tp D_{21} & 0
	\end{array}\right]
\]
After substituting $\mathcal{F} = \mathcal{G}^\textup{est}_1 +
\bar{\mathcal{F}}$, computing the coprime factorization, and changing
to more convenient coordinates, the resulting unconstrained
optimization problem is
\begin{align*}
  \textup{minimize}&\quad\left\|
  \left[\begin{array}{cc|c}
    \hat A & (\hat L - L)C_2 & (\hat L-L) D_{21} \\
    0 & A_L & B_2 + L D_{21} \\ \hlinet
    I & I & 0 \\
    I & 0 & 0 \\
    K & 0 & 0
    \end{array}\right]  -\mathcal{Q}
  \mathcal{D}\,
\right\|_2 \\
\textup{subject to}&\quad \mathcal{Q}\in\RHtwo
\end{align*}
where $\hat A$ is defined in~\eqref{a:ahat},
and
\[
\mathcal{D} =    \left[\begin{array}{cc|c}
        \hat A & (\hat L - L)C_2 & (\hat L-L) D_{211} \\
        0 & A_L & B_2 + L D_{21} \\ \hlinet
        E_1^\tp C_2  & E_1^\tp C_2 & E_1^\tp D_{21} 
        \end{array}\right] 
\]
Optimality of $\mathcal{Q}=0$ is established by verifying the
analogous orthogonality relationship to~\eqref{orth}. This is easily
done using the Gramian identity from Theorem~\ref{thm:diag_gram}. Then
we have
\[
\mathcal{G}^\textup{est} = \left[\begin{array}{c|cc}
\hat A & -\hat L E_1 & B_2 \\ \hlinet
I & 0 & 0\\
I & 0 & 0\\
K & 0 & 0
\end{array}\right]
\]
Comparing with~\eqref{e1}, we can see that
\[
\bmat{ \zeta \\ \zeta \\ \hat u }
=
\mathcal{G}^\text{est}\bmat{y_1 \\ u_\zeta} 
\]
and the result follows.
\end{proof}

\noindent
The result of Theorem~\ref{thm:structure_tpof} that
$
\zeta = x_{|y_1, u_\zeta} 
$
has a clear interpretation that $\zeta$ is a Kalman filter of $x$.
However, this result assumes that Player~2 is
implementing the optimal policy for $u_2$. This is important because
the optimal estimator gain depends explicitly on this choice of
policy. In contrast, in the centralized case,
the optimal estimation gain does not depend on the choice
of control policy.
  
  Note that using $u_1$ or $\hat u$ instead of $u_\zeta$ in
  Theorem~\ref{thm:structure_tpof} yields an incorrect
  interpretation of $\zeta$. In these coordinates,
  \begin{align*}
  \zeta &= \left[\begin{array}{c|cc}
  A+B_2E_2E_2^\tp K +\hat L C_2 & -\hat LE_1 & B_2E_1 \\ \hlinet
  I & 0 & 0\end{array}\right]
  \bmat{y_1 \\ u_1}\,\,\,\text{or}
\\
  \zeta &= \left[\begin{array}{c|cc}
  A+B_2 K +\hat L C_2 & -\hat LE_1 & B_2 \\ \hlinet
  I & 0 & 0\end{array}\right]
  \bmat{y_1 \\ \hat u}
  \end{align*}
In general, neither of these maps is stable, so for these choices of
signals, $\zeta$ cannot be an estimator of $x$ in the sense of
Definitions~\ref{def:kfc} and~\ref{def:must}.

We can now state the optimal two-player controller in the following simple form.
\begin{equation}
  \label{eqn:hj}
  u  = K x_{|y_1,u_\zeta} + 
  \bmat{0 & 0 \\ H & J} 
  \bl(
  x_{|y,u}
  - x_{|y_1,u_\zeta}
  \br)
\end{equation}
where $H\defeq\hat{K}_{21}$. These equations make clear that Player~1 is
using the same control action that it would use if the controller was
centralized with both players only measuring $y_1$.  We can also see
that the control action of Player~2 has an additional correction term,
given by the gain matrix $\bmat{H & J}$ multiplied by the difference
between the players' estimates of the states. Note that $u_\zeta$ is a
function only of $y_1$, as can be seen from the state-space equations
\begin{align*}
  \dot\zeta &= A_K \zeta - \hat{L}E_1\bl(y_1 - \bmat{C_{11} & 0 } \zeta\br)\\
  u_\zeta &= (K-\hat{K}) \zeta 
\end{align*}
The orthogonality relationships of the form~\eqref{orth} used in
Lemma~\ref{lem:structure_centralized} and
Theorem~\ref{thm:structure_tpof} to solve the $\Htwo$ optimization
problems may also be interpreted in terms of signals in the optimally
controlled closed-loop.

\begin{cor}\label{cor:orth}
Suppose $u,w,x$ and $y$ satisfy~\eqref{eq:ssplant}--\eqref{e3}.
Then  the maps $E_2: w\to x - \xi$ and $R_2: w\to y - C_2\xi$,
which give the error and residual for Player~2, are
\begin{align*}
E_2 &= \stsp{A_L}{B_1+LD_{21}}{I}{0} &
R_2 &= \stsp{A_L}{B_1+LD_{21}}{C_2}{D_{21}}
\end{align*}
The maps $E_1: w\to x - \zeta$ and $R_1: w \to y_1 -
C_{11}\zeta_1 $, which give the error and residual for~Player~1,  are
\begin{align*}
E_1  &= \left[\begin{array}{cc|c} \hat A & (\hat L-L) C_2 & (\hat L-L) D_{21} \\ 
    0 & A_L & B_1+LD_{21} \\ \hlinet I&I&0\end{array}\right] \if\MODE2\\\else&\fi
R_1 &= \stsp{A_M}{ME_1^\tp - E_1^\tp L}{C_{11}}{E_1^\tp} R_2
\end{align*}
Furthermore,  the orthogonality conditions $E_2 R_2^* \in \Htwo^\perp$ and 
$E_1 R_1^* \in \Htwo^\perp$
are satisfied.
\end{cor}

As with Theorem~\ref{thm:diag_gram}, Corollary~\ref{cor:orth} lends
itself to a statistical interpretation. If $w$ is white Gaussian  
noise with unit intensity, and we consider the steady-state distributions of the error and
residual for the second player, $x-\xi$ and $y-C_2\xi$ respectively,
then they are independent. Similarly, the error and residual for the
first player, $x-\zeta$ and $y_1-C_{11}\zeta_1$, are also independent.


\subsection{Optimal cost}

We now compute the cost associated with the optimal control policy for
the two-player output-feedback problem. From centralized~$\Htwo$ theory~\cite{zdg},
there are many equivalent expressions for the optimal centralized cost. In particular,
\if\MODE2
\begin{align*}
	\norm{F_\ell(\mathcal{P},\mathcal{K}_\textup{cen})}_2^2 &=
	\normmmmm{\stsp{A_K}{B_1}{C_1+D_{12}K}{0}}_2^2 \\
	    &\hspace{12mm} + \normmmmm{\stsp{A_L}{B_1+LD_{21}}{D_{12}K}{0}}_2^2 \\
	&= \tr(XW) + \tr(YK^\tp RK) \\
	&= \tr(YQ) + \tr(XLVL^\tp)
	\end{align*}
\else
	\begin{align*}
	\norm{F_\ell(\mathcal{P},\mathcal{K}_\textup{cen})}_2^2 &=
	\normmmmm{\stsp{A_K}{B_1}{C_1+D_{12}K}{0}}_2^2
	    + \normmmmm{\stsp{A_L}{B_1+LD_{21}}{D_{12}K}{0}}_2^2 \\
	&= \tr(XW) + \tr(YK^\tp RK) \\
	&= \tr(YQ) + \tr(XLVL^\tp)
	\end{align*}
\fi
where $X,Y,K,L$ are defined in~\eqref{tp:ares}. 
Of course, the cost of the optimal two-player controller will be greater, so we have
\[
\norm{F_\ell(\mathcal{P}, \mathcal{K}_\textup{opt})}_2^2 = 
\norm{F_\ell(\mathcal{P},\mathcal{K}_\textup{cen})}_2^2 + \Delta
\]
where $\Delta \ge 0$ is the additional cost incurred by decentralization. We now give some
closed-form expressions for~$\Delta$ that are similar to the centralized formulae above.

\begin{thm}\label{thm:cost}

  The additional cost incurred by the optimal controller~\eqref{eqn:kopt}
  for the two-player problem~\eqref{opt:2p_output_feedback} as compared to the cost of 
  the optimal centralized controller is
  \begin{align*}
    \Delta &= \normmmmm{\stsp{\hat A}{(\hat L-L)D_{21}}{D_{12}(\hat K-K)}{0}}_2^2 \\
	&=\tr(\hat Y-Y)(\hat K-K)^\tp R(\hat K-K ) \\
	&=\tr(\hat X-X)(\hat L-L)V(\hat L-L)^\tp
  \end{align*}
  where $K$, $\hat K$, $L$, $\hat L$ are defined in
  \eqref{tp:ares}--\eqref{tp:gains}, $\hat X$ and $\hat Y$ are defined
  in~\eqref{ee:lyapY} and~\eqref{ee:lyapX}, and $\hat A$ is defined in
  \eqref{a:ahat}.
\end{thm}
\begin{proof}
  The key is to view $\mathcal{K}_\textup{opt}$ as a sub-optimal
  centralized controller. Centralized~$\Htwo$ theory
  \cite{zdg} then implies that
  \begin{equation}
    \label{eqn:opt2:cost_norm_stuff}
  \Delta = \norm{D_{12} \mathcal{Q}_\text{you}D_{21}}_2^2
  \end{equation}
  where $\mathcal{Q}_\text{you}$ is  the centralized Youla
  parameter. Specifically, $\mathcal{Q}_\text{you} =
  F_u(\mathcal{J}^{-1},\mathcal{K}_\textup{opt})$ and
  \[
  \mathcal{J}^{-1} = \left[\begin{array}{c|cc} A & B_2 & -L \\
      \hlinet
      C_2 & 0 & I \\
      -K & I & 0
    \end{array}\right]
  \]
  This centralized Youla parameterization contains the gains $K$ and $L$ instead of $K_d$ and $L_d$. After simplifying, we obtain
  \begin{equation}\label{Qs}
    \mathcal{Q}_\text{you} = \left[\begin{array}{c|c}
        \hat A & \hat L-L \\ \hlinet \hat K-K & 0
      \end{array}\right]
  \end{equation}
  substituting~\eqref{Qs} into~\eqref{eqn:opt2:cost_norm_stuff} yields the first formula.
  The second formula follows from evaluating~\eqref{eqn:opt2:cost_norm_stuff} in a different way.
  Note that $\norm{D_s + C_s(sI-A_s)^{-1}B_s}^2= \tr(C_sW_c C_s^\tp)$,  where
  $W_c$ is the associated controllability Gramian, given by
  \hbox{$A_sW_c + W_cA_s^\tp + B_sB_s^\tp = 0$}.  In the case of~\eqref{eqn:opt2:cost_norm_stuff}, the 
  controllability Gramian equation is precisely~\eqref{ee:lyapY},
  and therefore~\hbox{$W_c = \hat Y - Y$} and the second formula follows. 
  The third formula follows using the observability
  Gramian $W_o = \hat X - X$ given by~\eqref{ee:lyapX}.
\end{proof}

Theorem~\ref{thm:cost} precisely quantifies the cost
of decentralization. We note that $\Delta$ is small whenever $(\hat L -L)$ or $(\hat K -K)$ is small. Examining these cases individually, we find the following. If $L = \hat L$, then even the optimal centralized controller makes no use of $y_2$; it only adds noise and no new information. In the two-player context, this leads to $\zeta=\xi$, so both players have the same estimate of the global state, based only on $y_1$, and either player is capable of determining $u$.
On the other hand, if $K = \hat K$, then even the optimal centralized controller makes no use of $u_1$. In the two-player context, this leads to a policy in which Player~1 does nothing.


\subsection{Some special cases}

The optimal controller~\eqref{eqn:kopt} is given
by~\eqref{eqn:hj}, which we repeat here for convenience.
\[
  u  = K x_{|y_1,u_\zeta} + 
  \bmat{0 & 0 \\ H & J} 
  \bl(
  x_{|y,u}
  - x_{|y_1,u_\zeta}
  \br)
\]
Recall that $K$ and $J$ are found by solving standard AREs~\eqref{tp:ares}.
The coupling between estimation and control appears in the term
$H = -R_{22}^{-1}(B_{22}^\tp \Phi + S_{12}^\tp)$, which is found by solving the
coupled linear equations~\eqref{tp:phi}--\eqref{tp:psi}.

Several special cases of the two-player output-feedback problem have
previously been solved. We will see that in each case, the first component of
$(x_{|y,u}- x_{|y_1,u_\zeta})$ is zero. In other words, both players maintain
identical estimates of $x_1$ given their respective information.
Consequently, $u$ no longer depends on $H$ and there is no need to
compute $\Phi$ and $\Psi$. We now examine these special cases in more detail.

\paragraph{Centralized} The problem becomes centralized when both players
have access to the same information. In this case, both players maintain
identical estimates of the entire state~$x$. Thus,
$x_{|y,u} = x_{|y_1,u_\zeta} = \hat x$
and we recover the well-known centralized solution~\hbox{$u = K \hat x$}.

\paragraph{State feedback}  The state-feedback problem for two players is the
case where our measurements are noise-free, so that~\hbox{$y_1 = x_1$} and $y_2 =
x_2$.  Therefore, both players measure $x_1$ exactly.  This case is
solved in~\cite{shah10,swigart10} and the solution takes the following
form, which agrees with our general formula.
\[
u = K\bmat{x_1 \\ \hat x_{2|1}} + \bmat{0 \\ J}(x_2 - \hat x_{2|1})
\]
Here, $\hat x_{2|1}$ is an estimate of $x_2$ given the information
available to Player~1, as stated in~\cite{swigart10}.

\paragraph{Partial output feedback}  In the partial output-feedback case,
$y_1=x_1$ as in the state-feedback case, but $y_2$ is a noisy linear
measurement of both states. This case is solved in~\cite{swigart_partial} and the solution takes the following form (using notation
from~\cite{swigart_partial}), which agrees with our general formula.
\[
u = K\bmat{x_1 \\ \hat x_{2|1}} + \bmat{0\\J}(\hat x_{2|2}-\hat x_{2|1})
\]

\paragraph{Dynamically decoupled}  In the dynamically decoupled case, all
measurements are noisy, but the dynamics of both systems are
decoupled.  This amounts to the case where~$A_{21}=0$, $B_{21}=0$,
$C_{21}=0$, and $W$, $V$, $U$ are block-diagonal. Due to the
decoupled dynamics, the estimate of $x_1$ based on $y_1$ does not
improve when additionally using $y_2$. This case is solved
in~\cite{jonghan} and the solution takes the following form, which
agrees with our general formula.
\[
u = K\bmat{\hat x_{1|1} \\ \hat x_{2|1}} + \bmat{0\\J}(\hat x_{2|2}-\hat x_{2|1})
\]
Note that estimation and control are decoupled in all the special cases examined above.
This fact allows the optimal controller to be computed by merely solving some
subset of the AREs~\eqref{tp:ares}. In the general case
however, estimation and control are coupled via $\Phi$ and~$\Psi$ 
in~\hbox{\eqref{tp:phi}--\eqref{tp:psi}}.


\section{Proofs}
\label{sec:main_proof}

\subsection{Existence and uniqueness of the controller}
\label{sec:origins_a}

In order to prove existence and uniqueness, our general approach
is to first convert the optimal control problem into a model-matching
problem using Corollary~\ref{cor:optimization_model_match}.  This
model-matching problem has stable parameters $\mathcal{T}_{ij}$, and
so may be solved using standard optimization methods on the Hilbert
space $\Htwo$. We therefore turn our attention to this class of
problems.  We will need the following assumptions.
\begin{enumerate}[\it{B}1)]
\itemsep=1.5mm
\item $\mathcal{T}_{11}(\infty) = 0$  \label{ass:Bfirst}
\item $\mathcal{T}_{12}(\i\omega)$ has full column rank for all $\omega\in\R\cup\{\infty\}$ \label{ass:Bmid}
\item $\mathcal{T}_{21}(\i\omega)$ has full row rank for all $\omega\in\R\cup\{\infty\}$\label{ass:Blast}
\end{enumerate}
The optimality condition for centralized model-matching is given in
the following lemma.
\begin{lem}
\label{lem:optimality_conditions_centralized}
Suppose $\mathcal{T}\!\in\RHinf$ satisfies
Assumptions~B\ref{ass:Bfirst}--B\ref{ass:Blast}. Then the model-matching problem
 \begin{equation}
       \label{opt:centralized_model_match}
       \begin{aligned}
         \textup{minimize}\qquad
         & \bigl\| \mathcal{T}_{11} 
         + \mathcal{T}_{12} \mathcal{Q} \mathcal{T}_{21} \bigr\|_2 \\
         \textup{subject to}\qquad& \mathcal{Q} \in \Htwo
       \end{aligned}
     \end{equation}
     has a unique solution. Furthermore, $\mathcal{Q}$ is the
     minimizer of~\eqref{opt:centralized_model_match} if and only
     if
\begin{equation}
  \label{eqn:copt}
  \mathcal{T}_{12}^*\left(
    \mathcal{T}_{11} + \mathcal{T}_{12} \mathcal{Q} \mathcal{T}_{21}
  \right) \mathcal{T}_{21}^* \in \Htwo^\perp
\end{equation}
\end{lem}

\begin{proof}
The Hilbert projection theorem (see, for example~\cite{luenberger})
states that if $H$ is a Hilbert space, $S\subseteq H$
is a closed subspace, and $b\in H$, then there exists a unique $x\in S$
that minimizes $\norm{x-b }_2$. Furthermore, a necessary and
sufficient condition for optimality of $x$ is that $x\in S$ and $x-b
\in S^\perp$.  Given a bounded linear map $A:H \rightarrow H$ which is
bounded below, define $T \defeq \{ Ax\,\vert\, x\in S\}$. Then $T$ is closed
and the projection theorem implies that the problem
\begin{align*}
  \textup{minimize} \qquad & \norm{Ax-b}_2\\
  \textup{subject to}\qquad  & x \in S
\end{align*}
has a unique solution. Furthermore, $x$ is optimal if and only if
$x\in S$ and $A^*(Ax-b)\in S^\perp$.  This result directly implies the
lemma, by setting $H \defeq \Ltwo$ and $S \defeq \Htwo$, and defining
the bounded linear operator~$A:\Ltwo\to \Ltwo$ by
$A\mathcal{Q}\defeq\mathcal{T}_{12}\mathcal{Q}\mathcal{T}_{21}$, with
adjoint
$A^*\mathcal{P}\defeq\mathcal{T}_{12}^*\mathcal{P}\mathcal{T}_{21}^*$.
The operator $A$ is bounded below as a consequence of
Assumptions~B\ref{ass:Bmid}--B\ref{ass:Blast}.
\end{proof}

We now develop an optimality condition similar to~\eqref{eqn:copt},
but for the two-player structured model-matching problem. 

\begin{lem}\label{lem:optimality_conditions_2p}   
  Suppose $\mathcal{T}\!\in\RHinf$ satisfies
    Assumptions~B\ref{ass:Bfirst}--B\ref{ass:Blast}. Then the
  two-player model-matching problem
  \begin{equation}
    \label{opt:2p_model_match}
    \begin{aligned}
      \textup{minimize}\qquad& \normm{
        \mathcal{T}_{11} + \mathcal{T}_{12} \mathcal{Q}\mathcal{T}_{21}
      }_2 \\
      \textup{subject to}\qquad& \mathcal{Q} \in \Lower(\Htwo)
    \end{aligned}
  \end{equation}
  has a unique solution. Furthermore,
  $\mathcal{Q}$ is the minimizer of~\eqref{opt:2p_model_match} if and only if
  \begin{equation}\label{optcond_2p}
    \mathcal{T}_{12}^*\left(
    \mathcal{T}_{11} + \mathcal{T}_{12} \mathcal{Q} \mathcal{T}_{21}
    \right) \mathcal{T}_{21}^* \in
    \bmat{ \Htwo^\perp & \Ltwo \\ \Htwo^\perp & \Htwo^\perp }
  \end{equation}
\end{lem}
\begin{proof}
  The proof follows exactly that of Lemma~\ref{lem:optimality_conditions_centralized}.
\end{proof}

Each of these model-matching problems has a rational transfer function solution,
as we now show.

\begin{lem}\label{lem:centralized_model_matching}
Suppose $\mathcal{T}\in\RHinf$ satisfies
Assumptions~B\ref{ass:Bfirst}--B\ref{ass:Blast} and has a minimal
joint realization given by
\begin{equation}
  \label{eqn:t}
  \bmat{ \mathcal{T}_{11} & \mathcal{T}_{12} \\
    \mathcal{T}_{21} & 0 } =
  \left[\begin{array}{c|cc}
      \bar A & \bar B_1 & \bar B_2 \\ \hlinet
      \bar C_1 & 0 & \bar D_{12} \\
      \bar C_2 & \bar D_{21} & 0
    \end{array}\right]
\quad 
\if\MODE2
	\begin{array}{l}\text{where $\bar A$ is}\\\text{Hurwitz.}\end{array}
\else
	\text{where $\bar A$ is Hurwitz.}
\fi
\end{equation}
The solution to the centralized model-matching problem~\eqref{opt:centralized_model_match} is
rational, and has realization
\begin{equation}\label{eq:Qoptcen}
\mathcal{Q}_\textup{opt} = \left[\begin{array}{cc|c}
 \bar A_K & \bar B_2 \bar K & 0 \\
 0 & \bar A_L & -\bar L \\ \hlinet
 \bar K & \bar K & 0
\end{array}\right]
\end{equation}
where $\bar K$, $\bar L$ are defined in \eqref{tp:ares}, and $\bar A_K$, $\bar A_L$ are defined in~\eqref{a:ahat} with all state-space parameters replaced by their barred counterparts.
\end{lem}

\begin{proof}
It is straightforward to verify that Assumptions
A\ref{ass:Afirst}--A\ref{ass:Alast} hold for the
realization~\eqref{eqn:t}. Optimality is verified by
substituting~\eqref{eq:Qoptcen} directly into the optimality
condition~\eqref{eqn:copt} and applying
Lemma~\ref{lem:optimality_conditions_centralized}.
\end{proof}

\begin{lem}\label{lem:rationality_2p}   
  Suppose $\mathcal{T}\in\RHinf$ satisfies Assumptions~B\ref{ass:Bfirst}--B\ref{ass:Blast}. Then the optimal solution
  of the two-player model-matching problem~\eqref{opt:2p_model_match}
  is rational.
\end{lem}

\begin{proof}
  Since the $\mathcal{H}_2$-norm is invariant under rearrangement of
  matrix elements, we may vectorize~\cite{horn_1994} the contents of
  the norm in \eqref{opt:2p_model_match} to obtain
  \begin{equation}\label{asdf4}
    \begin{aligned}
      \textup{minimize}\qquad& \normm{ \vecc(\mathcal{T}_{11}) +
      (\mathcal{T}_{21}^\tp \otimes 
        \mathcal{T}_{12}) \vecc(\mathcal{Q})}_2 \\
      \textup{subject to}\qquad& \mathcal{Q} \in \Lower(\Htwo)
    \end{aligned}
  \end{equation}
  Due to the sparsity pattern of $\mathcal{Q}$, some entries of
  $\vecc(\mathcal{Q})$ will be zero. Let $E$ be the identity matrix
  with columns removed corresponding to these zero-entries.  Then
  $\mathcal{Q}\in\Lower(\Htwo)$ if and only if $\vecc(\mathcal{Q}) = E q$ 
  for some  $q \in \Htwo$. Then~\eqref{asdf4} is 
  equivalent to
  \begin{equation}
    \label{asdf3}
    \begin{aligned}
      \textup{minimize}\qquad
      & \normm{\vecc(\mathcal{T}_{11}) +
       (\mathcal{T}_{21}^\tp \otimes 
        \mathcal{T}_{12})E q}_2 \\
      \textup{subject to}\qquad& q \in \Htwo
    \end{aligned}
  \end{equation}
  This is a centralized model-matching problem of the form
  \eqref{opt:centralized_model_match}. Using
  standard properties of the Kronecker product, one may verify that
  Assumptions~B\ref{ass:Bfirst}--B\ref{ass:Blast} are inherited by
    \begin{equation}
    \label{kronkron}
    \bmat{\vecc(\mathcal{T}_{11}) &
      (\mathcal{T}_{21}^\tp \otimes \mathcal{T}_{12})E \\
      1 & 0}
    \end{equation}
    By Lemma~\ref{lem:centralized_model_matching} the 
    optimal $q$ is rational, and hence the result follows.
\end{proof}


We now prove that under the assumptions of Theorem~\ref{thm:main}, the
two-player output-feedback problem~\eqref{opt:2p_output_feedback} has
a unique solution.

\medskip
\begin{proofe}{Proof of Theorem~\ref{thm:main}, Part (i).}
  We apply Corollary~\ref{cor:optimization_model_match} to reduce the
  optimal control problem to the two-player model-matching problem
  over $\RHinf$.  It is straightforward to check that under
  Assumptions A\ref{ass:Afirst}--A\ref{ass:Alast}, the particular
  $\mathcal{T}$ given in~\eqref{eqn:T} satisfies
  B\ref{ass:Bfirst}--B\ref{ass:Blast}.  These assumptions imply that
  $\mathcal{T}_{12}(\infty)$ is left-invertible and
  $\mathcal{T}_{21}(\infty)$ is right-invertible. Thus, the optimal
  solution $\mathcal{Q}_\textup{opt}\in\RHinf$ must have
  $\mathcal{Q}(\infty)=0$ to ensure that $\| \mathcal{T}_{11} +
  \mathcal{T}_{12} \mathcal{Q} \mathcal{T}_{21} \|_2$ is finite.
  Therefore we may replace the constraint that $\mathcal{Q} \in
  \Lower(\RHinf)$ as in Corollary~\ref{cor:optimization_model_match}
  with the constraint that $\mathcal{Q}\in\Lower(\RHtwo)$ without any
  loss of generality. Existence and uniqueness now follows from
  Lemma~\ref{lem:optimality_conditions_2p}.  Rationality follows from
  Lemma~\ref{lem:rationality_2p}.
\end{proofe}
The vectorization approach of Lemma~\ref{lem:optimality_conditions_2p}
effectively reduces the structured model-matching
problem~\eqref{opt:2p_model_match} to a centralized model-matching
problem, which has a known solution, shown later in
Lemma~\ref{lem:centralized_model_matching}.  Unfortunately,
constructing the solution in this manner is not feasible in practice
because it requires finding a state-space realization of the Kronecker
system~\eqref{kronkron}. This leads to a dramatic increase in state
dimension, and requires solving a large Riccati equation. Furthermore,
we lose any physical interpretation of the states, as mentioned in
Section~\ref{sec:intro}.


\subsection{Formulae for the optimal controller}

\begin{proofe}{Proof of Theorem~\ref{thm:main}, Part (ii).}
  Suppose the linear equations~\eqref{tp:phi}--\eqref{tp:psi} are satisfied by some $\Phi$, $\Psi$ and the proposed~$\mathcal{K}_\textup{opt}$ has been defined according to~\eqref{eqn:kopt}. We will make the following simplifying assumption.
    \begin{equation}\label{ass:LK}
    L_1 = M\qquad\text{and}\qquad K_2 = J
    \end{equation}
  where $L_1$ and $K_2$ were originally defined in Lemma~\ref{lem:conditions_stabilizability}, and $M$ and $J$ are defined in~\eqref{tp:ares}.
    There is no loss of generality in choosing this particular parameterization for~$\mathcal{T}$, but it leads to simpler algebra.
To verify optimality of
  $\mathcal{K}_\textup{opt}$, we use the parameterization of
  Theorem~\ref{thm:stabilizing_controllers} to find the
  $\mathcal{Q}_\textup{opt}$ that generates
  $\mathcal{K}_\textup{opt}$. The computation yields
\begin{align}
\label{Qopt}\notag
  \mathcal{Q}_\textup{opt} &= F_u(\mathcal{J}_d^{-1},\mathcal{K}_\textup{opt}) \\&=
  \left[\begin{array}{ccc|c}
      A_K & -\hat L C_2 & 0  & \hat L \\
      0 & \hat A & -B_2\hat K &  L_d - \hat L\\
      0 & 0 & A_L & L_d- L \\ \hlinet
      K_d - K & K_d-\hat K & \hat K & 0
    \end{array}\right] \notag\\ 
    &= \left[\begin{array}{cc|c}
          A_K & 0  & \hat L \\
          0 & A_L & L_d - L \\ \hlinet
          K_d - K & \hat K & 0
        \end{array}\right]
\end{align}
where $A_K$, $A_L$, $\hat A$ are defined in~\eqref{a:ahat}. The last simplification in~\eqref{Qopt} comes thanks to~\eqref{ass:LK}.
The sparsity structure of the gains $\hat L$ and $\hat K$
gives $\mathcal{Q}_\textup{opt}$ a block-lower-triangular
structure. Note also that $A_K$, $A_L$, $\hat A$ are Hurwitz, so $\mathcal{Q}_\textup{opt} \in \Lower(\RHtwo)$. It follows from Theorem~\ref{thm:stabilizing_controllers} that $\mathcal{K}_\textup{opt}$ is an admissible stabilizing controller.

We now directly verify that~$\mathcal{Q}_\textup{opt}$ defined in~\eqref{Qopt} satisfies the optimality condition~\eqref{optcond_2p} for the model-matching problem characterized by the $\mathcal{T}$ given in~\eqref{eqn:T}.
The closed-loop map $\mathcal{T}_{11} + \mathcal{T}_{12} \mathcal{Q}_\textup{opt} \mathcal{T}_{21}$ has a particularly nice expression. Substituting in \eqref{Qopt} and \eqref{eqn:T} and simplifying, we obtain
\if\MODE2
\begin{align*}
\mathcal{T}_{11} + \mathcal{T}_{12} \mathcal{Q}_\textup{opt} \mathcal{T}_{21}
&= \stsp{A_\cl }{ B_\cl }{ C_\cl }{0}
\\
&\hspace{-25mm}=\left[\begin{array}{ccc|c}
A_K & -\hat L C_2 & 0 & -\hat L D_{21} \\
0 & \hat A & -B_2\hat K & B_1+\hat L D_{21} \\
0 & 0 & A_L & B_1 + L D_{21} \\ \hlinet
C_1+D_{12}K & C_1+D_{12}\hat K & -D_{12}\hat K & 0
\end{array}\right]
\end{align*}
\else
\begin{align*}
\mathcal{T}_{11} + \mathcal{T}_{12} \mathcal{Q}_\textup{opt} \mathcal{T}_{21}
&= \stsp{A_\cl }{ B_\cl }{ C_\cl }{0}
\\
&=\left[\begin{array}{ccc|c}
A_K & -\hat L C_2 & 0 & -\hat L D_{21} \\
0 & \hat A & -B_2\hat K & B_1+\hat L D_{21} \\
0 & 0 & A_L & B_1 + L D_{21} \\ \hlinet
C_1+D_{12}K & C_1+D_{12}\hat K & -D_{12}\hat K & 0
\end{array}\right]
\end{align*}
\fi
Note that $K_d$ and $L_d$ are now absent, as the optimal closed-loop map does not depend on the choice of parameterization. The left-hand side of the optimality condition~\eqref{optcond_2p} is therefore
\if\MODE2
\begin{multline*}
  \mathcal{T}_{12}^*\left(
    \mathcal{T}_{11} + \mathcal{T}_{12} \mathcal{Q}_\textup{opt} \mathcal{T}_{21}
  \right) \mathcal{T}_{21}^* \\
  = \left[\begin{array}{ccc|c}
  -A_{Kd}^\tp & -C_{Kd}^\tp C_\cl & 0 & 0 \\
  0 & A_\cl & B_\cl B_{Ld}^\tp & B_\cl D_{21}^\tp \\ 
  0 & 0 & -A_{Ld}^\tp & -C_2^\tp \\ \hlinet
  B_2^\tp & D_{12}^\tp C_\cl & 0 & 0
  \end{array}\right]
\end{multline*}
\else
\begin{equation*}
  \mathcal{T}_{12}^*\left(
    \mathcal{T}_{11} + \mathcal{T}_{12} \mathcal{Q}_\textup{opt} \mathcal{T}_{21}
  \right) \mathcal{T}_{21}^* 
  = \left[\begin{array}{ccc|c}
  -A_{Kd}^\tp & -C_{Kd}^\tp C_\cl & 0 & 0 \\
  0 & A_\cl & B_\cl B_{Ld}^\tp & B_\cl D_{21}^\tp \\ 
  0 & 0 & -A_{Ld}^\tp & -C_2^\tp \\ \hlinet
  B_2^\tp & D_{12}^\tp C_\cl & 0 & 0
  \end{array}\right]
\end{equation*}
\fi
Apply Lemmas~\ref{lem:lyapY} and~\ref{lem:lyapX} to define~$\hat X$ and~$\hat Y$. Now perform the state transformation $x\mapsto Tx$ with
\[
T \defeq \bmat{ I & -\bmat{X & \hat X & 0} & 0 \\
0 & I & -\bmat{0 \\ \hat Y \\ Y} \\
0 & 0 & I }
\]
At this point we make use of the Riccati equations~\eqref{tp:ares}, as
well as the Sylvester equations~\eqref{tp:phi}--\eqref{tp:psi} via the
identities~\eqref{ee:lyapY}--\eqref{eg1}. This leads to a state space
realization with $5n$ states, and sparsity pattern of the form
\[
\Omega = \left[\begin{array}{ccccc|c}
-A_{Kd}^\tp & 0 & E_1 \star & E_1 \star & \star & \star \\
0 & \star & \star & 0 & \star E_2^\tp & \star \\
0 & 0 & \hat A & \star & \star E_2^\tp & \star E_2^\tp  \\
 0 & 0 & 0 & \star & 0 & 0 \\
0 & 0 & 0 & 0 & -A_{Ld}^\tp & -C_2^\tp \\ \hlinet
B_2^\tp & 0 & E_1 \star & \star & \star & 0
\end{array}\right]
\]
where $\star$ denotes a matrix whose value is unimportant.
The second and fourth states are unobservable and uncontrollable, respectively. Removing them, we are left with
\[
\Omega = \left[\begin{array}{ccc|c}
-A_{Kd}^\tp & E_1 \star & \star & \star \\
0 & \hat A & \star E_2^\tp & \star E_2^\tp  \\
0 & 0 & -A_{Ld}^\tp & -C_2^\tp \\ \hlinet
B_2^\tp & E_1 \star & \star & 0
\end{array}\right]
\]
Because of the block-triangular structure of $A$, $B_2$, $C_2$ and the block-diagonal structure of $K_d$, $L_d$, it is straightforward to check that
\[
\Omega E_1 =
\left[\begin{array}{cc|c}
-A_{Kd}^\tp & \star & \star \\
0 & -A_M^\tp & -C_{11}^\tp \\ \hlinet
B_2^\tp & \star &  0
\end{array}\right] \in \Htwo^\perp
\]
Similarly, $E_2^\tp \Omega \in \Htwo^\perp$, and therefore
\[
  \mathcal{T}_{12}^*\left(
    \mathcal{T}_{11} + \mathcal{T}_{12} \mathcal{Q}_\textup{opt} \mathcal{T}_{21}
  \right) \mathcal{T}_{21}^* \in
  \bmat{ \Htwo^\perp & \Ltwo \\ \Htwo^\perp & \Htwo^\perp }
\]
So by Lemma~\ref{lem:optimality_conditions_2p},
$\mathcal{Q}_\textup{opt}$ is the solution to the model-matching
problem~\eqref{opt:2p_model_match}, and therefore must also be a
minimizer of~\eqref{opt:2pmm_gen}. It follows from
Corollary~\ref{cor:optimization_model_match} that
$\mathcal{K}_\textup{opt}$ is the unique optimal controller for the
two-player output-feedback problem~\eqref{opt:2p_output_feedback}.
\end{proofe}


\subsection{Existence of solutions to the Sylvester equations}
\label{sec:origins}

We now show that the linear equations~\eqref{tp:phi}--\eqref{tp:psi}
have a solution.  First, we show that existence of the optimal
controller implies a certain fixed point property.  A simple necessary
condition for optimality is person-by-person optimality; by fixing any
part of the optimal $\mathcal{Q}$ and optimizing over the rest, we
cannot improve upon the optimal cost.

\begin{lem}\label{lem:coupled_problems}  
  Suppose $\mathcal{T}$ is given by~\eqref{eqn:T} and Assumptions A\ref{ass:Afirst}--A\ref{ass:Alast} hold.  Further suppose that
  $\mathcal{Q}\in\Lower(\RHtwo)$ and
  \[
  \mathcal{Q} \defeq \bmat{\mathcal{Q}_{11} & 0 \\
    \mathcal{Q}_{21} & \mathcal{Q}_{22}}
  \]
  Define the following partial optimization functions.
  \begin{align}
  g_{1}(\mathcal{Q}_{11}) &\defeq \argmin_{\mathcal{Q}_{22}\in \RHtwo} \left( \min_{\mathcal{Q}_{21}\in \RHtwo} \normm{\mathcal{T}_{11} +  
        \mathcal{T}_{12} \mathcal{Q} \mathcal{T}_{21} }_2 \right) \label{opt:Q11fixed} \\
  g_{2}(\mathcal{Q}_{22}) &\defeq \argmin_{\mathcal{Q}_{11}\in \RHtwo} \left( \min_{\mathcal{Q}_{21}\in \RHtwo} \normm{\mathcal{T}_{11} +  
        \mathcal{T}_{12} \mathcal{Q} \mathcal{T}_{21} }_2 \right)\label{opt:Q22fixed}
  \end{align}
  If $\mathcal{Q}$ is optimal for the two-player model-matching problem~\eqref{opt:2p_model_match}, then
  \[
  g_2(g_1(\mathcal{Q}_{11}))=\mathcal{Q}_{11}
  \qquad\text{and}\qquad
  g_1(g_2(\mathcal{Q}_{22}))=\mathcal{Q}_{22}
  \]
\end{lem}

\begin{proof}
Under the stated assumptions, the optimization
problems~\eqref{opt:Q11fixed}--\eqref{opt:Q22fixed} have unique
optimal solutions.  If $\mathcal{Q}$ is optimal, then clearly we have
$\mathcal{Q}_{22}=g_1(\mathcal{Q}_{11})$ and $\mathcal{Q}_{11} =
g_2(\mathcal{Q}_{22})$. The result follows by substituting one
identity into the other.
\end{proof}

An alternative way of stating Lemma~\ref{lem:coupled_problems} is that the optimal $\mathcal{Q}_{11}$ and $\mathcal{Q}_{22}$ are the fixed points of the maps $g_2\circ g_1$ and $g_1\circ g_2$ respectively. 
Our next step is to solve these fixed-point equations analytically. The key insight is that \eqref{opt:Q11fixed}--\eqref{opt:Q22fixed} are centralized model-matching problems of the form~\eqref{opt:centralized_model_match}.
In the following lemma, we fix $\mathcal{Q}_{11}$ and we find $\bmat{\mathcal{Q}_{21} & \mathcal{Q}_{22}}$ that minimizes the right-hand side of~\eqref{opt:Q11fixed}.

  \begin{lem}
    \label{lem:fixed_Q11}   
    Assume   $\mathcal{T}\in\RHinf$ is given by~\eqref{eqn:T}. Suppose
    Assumptions A\ref{ass:Afirst}--A\ref{ass:Alast} hold together with the
    structural requirement~\eqref{a:tri_form} and the parameter choice~\eqref{ass:LK}. If
    \[
    \mathcal{Q}_{11} \defeq
    \stsp{A_P}{B_P}{C_P}{0}
    \qquad\text{where $A_P$ is Hurwitz}
    \]
    then the right-hand side of \eqref{opt:Q11fixed} is minimized by
    \if\MODE2
    \begin{multline}\label{eqn:Q21_Q22_opt}
        \bmat{\mathcal{Q}_{21} & \mathcal{Q}_{22}}
        =\\
        \addtolength{\arraycolsep}{-1pt}
        \left[\begin{array}{ccc|c}
            A_{Kd}+B_2\bmat{0 \\ \bar K_1} & B_2\bmat{0 \\ \bar K_2} & B_2\bmat{C_P \\ \bar K_3} & -L_d \\
            0 & A_L & 0 & L_d-L \\
            0 & 0& A_P & B_PE_1^\tp \\ \hlinet
            \bar K_1 & \bar K_2 & \bar K_3 & 0
          \end{array}\right]
    \end{multline}
    \else
    \begin{equation}\label{eqn:Q21_Q22_opt}
        \bmat{\mathcal{Q}_{21} & \mathcal{Q}_{22}}
        =
        \left[\begin{array}{ccc|c}
            A_{Kd}+B_2\bmat{0 \\ \bar K_1} & B_2\bmat{0 \\ \bar K_2} & B_2\bmat{C_P \\ \bar K_3} & -L_d \\
            0 & A_L & 0 & L_d-L \\
            0 & 0& A_P & B_PE_1^\tp \\ \hlinet
            \bar K_1 & \bar K_2 & \bar K_3 & 0
          \end{array}\right]
    \end{equation}
    \fi
    Here the quantities $\bar K_1$, $\bar K_2$, and $\bar K_3$ are defined by
\if\MODE2
\begin{equation}\label{eqn:Kbar}
\begin{aligned}
  \bar K_1 &= \bmat{ -R_{22}^{-1} \left( B_{22}^\tp \Theta_X + S_{12}^\tp + R_{12}^\tp K_1
    \right) & 0} \\ 
  \bar K_2 &= \bmat{ -R_{22}^{-1} \left( B_{22}^\tp \Phi + S_{12}^\tp \right) & J} \\
  \bar K_3 &= -R_{22}^{-1} \left( B_{22}^\tp \Gamma_X + R_{12}^\tp C_P \right)
\end{aligned}
\end{equation}
\else
\begin{equation}\label{eqn:Kbar}
\begin{gathered}
  \bar K_1 = \bmat{ -R_{22}^{-1} \left( B_{22}^\tp \Theta_X + S_{12}^\tp + R_{12}^\tp K_1
    \right) & 0} \qquad
  \bar K_2 = \bmat{ -R_{22}^{-1} \left( B_{22}^\tp \Phi + S_{12}^\tp \right) & J}
  \\
  \bar K_3 = -R_{22}^{-1} \left( B_{22}^\tp \Gamma_X + R_{12}^\tp C_P \right)
\end{gathered}
\end{equation}
\fi
where $\Theta_X$, $\Gamma_X$, and $\Phi$ are the unique solutions
to the linear equations
\if\MODE2
\begin{equation}
  \begin{aligned}\label{EQN:phi_final}
    A_J^\tp \Theta_X + \bmat{ \Theta_X & \tilde X} A_{Kd}E_1 + E_2^\tp C_{Kd}^\tp C_{Kd} E_1 &= 0
    \\
    A_J^\tp \Gamma_X + \Gamma_X A_P \hspace{48mm}&\\
    + \bl( \bmat{ \Theta_X & \tilde X}B_2E_1 + J^\tp R_{21} +
      S_{21} \br) C_P &= 0 \\
    A_J^\tp \Phi + \Phi A_M + 
    \tilde X A_{21} -\Theta_X M C_{11} \hspace{20mm}& \\
    +  J^\tp S_{12}^\tp + Q_{21} + \Gamma_X B_P C_{11} &= 0
  \end{aligned}
\end{equation}
\else
\begin{equation}
  \begin{aligned}\label{EQN:phi_final}
    A_J^\tp \Theta_X + \bmat{ \Theta_X & \tilde X} A_{Kd}E_1 + E_2^\tp C_{Kd}^\tp C_{Kd} E_1 &= 0
    \\
    A_J^\tp \Gamma_X + \Gamma_X A_P 
    + \bl( \bmat{ \Theta_X & \tilde X}B_2E_1 + J^\tp R_{21} +
      S_{21} \br) C_P &= 0 \\
    A_J^\tp \Phi + \Phi A_M + 
    \tilde X A_{21} -\Theta_X M C_{11} 
    +  J^\tp S_{12}^\tp + Q_{21} + \Gamma_X B_P C_{11} &= 0
  \end{aligned}
\end{equation}
\fi
where $Q,R,S$ are defined in~\eqref{a:QRWV},
   $\tilde X, Y, J, L$ are defined in~\eqref{tp:ares}, and $A_{Kd}, C_{Kd}$ are defined in~\eqref{defn:short}.
\end{lem}

\begin{proof}
Since $\mathcal{Q}_{11}$ is held fixed, group it with $\mathcal{T}_{11}$ to obtain
\if\MODE2
\begin{multline*}\label{harf}
  \mathcal{T}_{11} + \mathcal{T}_{12} \mathcal{Q} \mathcal{T}_{21}
  \\=  \left( \mathcal{T}_{11} + \mathcal{T}_{12} E_1\mathcal{Q}_{11} E_1^\tp
    \mathcal{T}_{21} \right) + \mathcal{T}_{12} E_2\bmat{ \mathcal{Q}_{21} &
    \mathcal{Q}_{22}} \mathcal{T}_{21}
\end{multline*}
\else
\begin{equation*}\label{harf}
  \mathcal{T}_{11} + \mathcal{T}_{12} \mathcal{Q} \mathcal{T}_{21}
  =  \left( \mathcal{T}_{11} + \mathcal{T}_{12} E_1\mathcal{Q}_{11} E_1^\tp
    \mathcal{T}_{21} \right) + \mathcal{T}_{12} E_2\bmat{ \mathcal{Q}_{21} &
    \mathcal{Q}_{22}} \mathcal{T}_{21}
\end{equation*}
\fi
This is an affine function of $\bmat{\mathcal{Q}_{21} & \mathcal{Q}_{22}}$, so the associated model-matching problem is centralized. Finding a joint realization of the blocks as in Lemma~\ref{lem:centralized_model_matching}, we obtain
  \begin{multline*}
    \addtolength{\arraycolsep}{-1pt}\bmat{\mathcal{T}_{11}\!+\!\mathcal{T}_{12}E_1\mathcal{Q}_{11}E_1^\tp
      \mathcal{T}_{21} & \mathcal{T}_{12}E_2 \\ \mathcal{T}_{21} & 0 }
    = \addtolength{\arraycolsep}{-1pt}\left[\begin{array}{c|cc} \bar A & \bar B_1 & \bar B_2 \\
            \hlinet
            \bar C_1 & 0 & \bar D_{12} \\
            \bar C_2 & \bar D_{21} & 0
          \end{array}\right] \\
    = \addtolength{\arraycolsep}{-1pt}\left[\begin{array}{ccc|cc}
        A_{Kd} & -L_dC_2 & B_2E_1C_P & -L_d D_{21} & B_2E_2 \\
        0 & A_{Ld} & 0 & B_{Ld} & 0 \\
        0 & B_PE_1^\tp C_2 & A_P & B_PE_1^\tp D_{21} & 0 \\ \hlinet
        C_{Kd} & C_1 & D_{12}E_1 C_P & 0 & D_{12}E_2 \\
        0 & C_2 & 0 & D_{21} & 0
      \end{array}\right] 
  \end{multline*}
  It is straightforward to check that
  Assumptions~B\ref{ass:Bfirst}--B\ref{ass:Blast} are satisfied for
  this augmented system.  Now, we may apply
  Lemma~\ref{lem:centralized_model_matching}. The result is that
  \begin{equation}
    \label{e:Q21Q22_first}
    \bmat{\mathcal{Q}_{21} & \mathcal{Q}_{22}}_\textup{opt} =
    \left[\begin{array}{cc|c}
        \bar A+\bar B_2\bar K & \bar B_2\bar K & 0 \\
        0 & \bar A+\bar L\bar C_2 & -\bar L \\ \hlinet
        \bar K  & \bar K & 0
      \end{array}\right]
  \end{equation}
  where we defined $(\bar X,\bar K) \defeq \are(\bar A,\bar B_2,\bar C_1,\bar D_{12})$
  and $(\bar Y,\bar L^\tp) \defeq \are(\bar A^\tp,\bar C_2^\tp,\bar
  B_1^\tp,\bar D_{21}^\tp)$. One can check that the stabilizing solution to the latter ARE is
  \[
  \bar Y = \bmat{0 & 0 & 0 \\ 0 & Y & 0 \\0 & 0 & 0 }
  \qquad\text{and}\qquad \bar L = \bmat{L_d\\ L-L_d \\ -B_PE_1^\tp}
  \]
  The former ARE is more complicated, however. Examining
  $\bar K = -\bar R^{-1}(\bar B^\tp \bar X + \bar S^\tp)$, we notice
  that due to all the zeros in $\bar B$, the only part of $\bar X$
  that affects the gain $\bar K$ is the second sub-row of the first
  block-row. In other words, if
  \[ 
  \bar X \defeq \bmat{ \bar X_{11} & \bar X_{12} & \bar X_{13} \\ \bar
    X_{21} & \bar X_{22} & \bar X_{23} \\ \bar X_{31} & \bar X_{32} &
    \bar X_{33} }
  \]
  then $\bar K$ only depends on $E_2^\tp \bar X_{11}$, $E_2^\tp \bar
  X_{12}$, and $E_2^\tp \bar X_{13}$. Multiplying the ARE for $\bar X$ on the left by $\bmat{ E_2^\tp & 0 & 0}$ and on the right by $\bmat{E_2^\tp & 0 & 0}^\tp$, we obtain the equation
\begin{equation*}
    A_J^\tp \tilde X + \tilde X A_J \\
    + (C_1E_2 + D_{12}E_2 J)^\tp(C_1E_2 + D_{12}E_2 J)
     = 0
    \end{equation*}
    where $\tilde X \defeq E_2^\tp \bar X_{11} E_2$. It is straightforward to see that $\tilde X$ as defined in~\eqref{tp:ares} satisfies this equation. Substituting this back into the ARE for $\bar X$ and multiplying on the left by $\bmat{ E_2^\tp & 0 & 0}$, we obtain
  \begin{multline}
    \label{_ugly1}
    A_J^\tp \addtolength{\arraycolsep}{-2pt}\bmat{ E_2^\tp \bar X_{11} & E_2^\tp \bar X_{12} &
      E_2^\tp \bar X_{13} } \\
      + \addtolength{\arraycolsep}{-2pt}\bmat{ E_2^\tp \bar X_{11} & E_2^\tp
      \bar X_{12} & E_2^\tp \bar X_{13} }
    \bmat{A_{Kd} & -L_d C_2 & B_2E_1C_P \\
      0 & A_{Ld} & 0 \\
      0 & B_PE_1^\tp C_2 & A_P} \\ 
     + (C_1E_2 + D_{12} E_2 J)^\tp \bmat{C_{Kd} & C_1 & D_{12}E_1 C_P} = 0
  \end{multline}
  Right-multiplying \eqref{_ugly1} by $\bmat{0 & E_2^\tp & 0}^\tp$, we conclude that $E_2^\tp \bar X_{12} E_2 = \tilde X$. Notice that~\eqref{_ugly1} is linear in the $\bar X$ terms. Assign the following names to the missing pieces
\begin{align*}
  E_2^\tp \bar X_{11} &\defeq \bmat{ \Theta_X & \tilde X } &
  E_2^\tp \bar X_{12} &\defeq \bmat{ \Phi & \tilde X } &
  E_2^\tp \bar X_{13} &\defeq \Gamma_X
\end{align*}
Upon substituting these definitions into \eqref{_ugly1} and 
simplifying, we obtain \eqref{EQN:phi_final}.  A similar substitution into the definition of $ \bar K = \bmat{\bar K_1 & \bar K_2 & \bar
  K_3} $ leads to the formulae~\eqref{eqn:Kbar}.
  
The equations \eqref{EQN:phi_final} have a unique solution. To see
why, note that they may be sequentially solved: for $\Omega_X$, then
for $\Gamma_X$, and finally for $\Phi$. Each is a Sylvester equation
of the form $A_1 \Omega + \Omega A_2 + A_0 = 0$, where $A_1$ and $A_2$
are Hurwitz, so the solution is unique. Furthermore, $\bar A + \bar
B \bar K$ is easily verified to be Hurwitz, so the procedure outlined
above produces the correct stabilizing $\bar X$.  Now, substitute
$\bar K$ and $\bar L$ into~\eqref{e:Q21Q22_first}. The result is a
very large state-space realization, but it can be greatly reduced by
eliminating uncontrollable and unobservable states. The result
is~\eqref{eqn:Q21_Q22_opt}. This reduction is not surprising, because
we solved a model-matching problem in which the joint realization for
the three blocks had a zero as the fourth block.
\end{proof}

We may solve \eqref{opt:Q22fixed} in a manner analogous to how we
solved~$\eqref{opt:Q11fixed}$. Namely, we can provide a formula for
the optimal~$\mathcal{Q}_{11}$ and $\mathcal{Q}_{21}$ as functions of
$\mathcal{Q}_{22}$. The result follows directly from
Lemma~\ref{lem:fixed_Q11} after we make a change of variables.

  \begin{lem}
    \label{lem:fixed_Q22}
    Assume   $\mathcal{T}\in\RHinf$ is given by~\eqref{eqn:T}.   Suppose
      Assumptions A\ref{ass:Afirst}--A\ref{ass:Alast} hold together with the
      structural requirement~\eqref{a:tri_form} and the parameter choice~\eqref{ass:LK}. If
    \[
    \mathcal{Q}_{22} \defeq \stsp{A_Q}{B_Q}{C_Q}{0}
    \qquad\text{where $A_Q$ is Hurwitz}
    \]
    then the right-hand side of \eqref{opt:Q22fixed} is minimized by
    \begin{equation}
      \label{eqn:Q11_Q21_opt}
      \if\MODE2\addtolength{\arraycolsep}{-2pt}\fi
      \begin{aligned}
        \bmat{\mathcal{Q}_{11} \\ \mathcal{Q}_{21}} \!=\!
        \left[\begin{array}{ccc|c}
            A_{Ld}+\bmat{\bar L_1 & 0}C_2 & 0 & 0 & \bar L_1 \\
            \bmat{\bar L_2 & 0} C_2 & A_K & 0 & \bar L_2 \\
            \bmat{ \bar L_3 & B_Q }C_2 & 0 & A_Q & \bar L_3 \\
            \hlinet -K_d & K_d-K & E_2 C_Q & 0
          \end{array}\right]
      \end{aligned}
    \end{equation}
    The quantities $\bar L_1$, $\bar L_2$, and $\bar L_3$ are defined by
\if\MODE2
    \begin{align*}
      \bar L_1 &= \bmat{ 0 \\ -(\Theta_YC_{11}^\tp+U_{12}^\tp)V_{11}^{-1}} \\
      \bar L_2 &= \bmat{ M \\ -(\Psi C_{11}^\tp+U_{12}^\tp)V_{11}^{-1}} \\
      \bar L_3 &= -\left( \Gamma_YC_{11}^\tp + B_QV_{21} \right)V_{11}^{-1}
    \end{align*}
\else
	\[
    \begin{gathered}
      \bar L_1 = \bmat{ 0 \\ -(\Theta_YC_{11}^\tp+U_{12}^\tp+L_2V_{12}^\tp)V_{11}^{-1}}
      \qquad
      \bar L_2 = \bmat{ M \\ -(\Psi C_{11}^\tp+U_{12}^\tp)V_{11}^{-1}} \\
      \bar L_3 = -\left( \Gamma_YC_{11}^\tp + B_QV_{21} \right)V_{11}^{-1}
    \end{gathered}
    \]
\fi
    where $\Theta_Y$, $\Gamma_Y$, and $\Psi$ are the unique solutions
    to the linear equations
\if\MODE2
    \begin{equation}
      \begin{aligned}\label{EQN:psi_final}
        E_2^\tp A_{Ld}\bmat{\tilde Y \\ \Theta_Y}
        + \Theta_YA_M^\tp
        + E_2^\tp B_{Ld}B_{Ld}^\tp E_1  &= 0 \\
        A_Q \Gamma_Y + \Gamma_Y A_M^\tp \hspace{46mm}&\\
        + B_Q\bbbl(E_2^\tp C_2 \bmat{\tilde Y \\ \Theta_Y} +
        V_{21}M^\tp + U_{21}\bbbr) &= 0
        \\
        A_J\Psi + \Psi A_M^\tp +  A_{21}\tilde Y - B_{22}J \Theta_Y \hspace{19mm}&\\
        + U_{12}^\tp M^\tp + W_{21} + B_{22}C_Q\Gamma_Y  &= 0
      \end{aligned}
    \end{equation}
\else
    \begin{equation}
      \begin{aligned}\label{EQN:psi_final}
        E_2^\tp A_{Ld}\bmat{\tilde Y \\ \Theta_Y}
        + \Theta_YA_M^\tp
        + E_2^\tp B_{Ld}B_{Ld}^\tp E_1  &= 0 \\
        A_Q \Gamma_Y + \Gamma_Y A_M^\tp 
        + B_Q\bbbl(E_2^\tp C_2 \bmat{\tilde Y \\ \Theta_Y} +
        V_{21}M^\tp + U_{21}\bbbr) &= 0
        \\
        A_J\Psi + \Psi A_M^\tp +  A_{21}\tilde Y - B_{22}J \Theta_Y
        + U_{12}^\tp M^\tp + W_{21} + B_{22}C_Q\Gamma_Y  &= 0
      \end{aligned}
    \end{equation}
\fi
    where $W, V, U$ are defined in~\eqref{a:QRWV},
   $X, \tilde Y, K, M$ are defined in~\eqref{tp:ares}, and $A_{Ld}, B_{Ld}$ are defined in~\eqref{defn:short}.
  \end{lem}

A key observation that greatly simplifies the extent to which the
optimization problems \eqref{opt:Q11fixed} and \eqref{opt:Q22fixed}
are coupled is that the $\mathcal{Q}_{ii}$ have simple state-space
representations.  In fact, we have the strong conclusion that for
any fixed~$\mathcal{Q}_{11}$, the optimal $\mathcal{Q}_{22} =
g_1(\mathcal{Q}_{11})$ has a fixed state dimension no greater than the
dimension of the plant, as in the following result.

\begin{thm}
\label{thm:simple_coupling}
Suppose $\mathcal{T}\in\RHinf$ is given by~\eqref{eqn:T} and Assumptions A\ref{ass:Afirst}--A\ref{ass:Alast} hold together with the structural requirement~\eqref{a:tri_form} and the parameter choice~\eqref{ass:LK}. The functions $g_i$ defined in~\eqref{opt:Q11fixed}--\eqref{opt:Q22fixed} are given by
\if\MODE2
    \begin{equation}\label{EQN:optQii}
    \begin{aligned}
      g_1(\mathcal{Q}_{11}) &=
      \stsp{A_L}{(L_d-L)E_2}{\bar K_2}{0}
      \\
      g_2(\mathcal{Q}_{22}) &=
      \stsp{A_K}{\bar L_2}{E_1^\tp(K_d-K)}{0}
    \end{aligned}
    \end{equation}
\else
    \begin{equation}\label{EQN:optQii}
    \begin{aligned}
      g_1(\mathcal{Q}_{11}) =
      \stsp{A_L}{(L_d-L)E_2}{\bar K_2}{0}
      \qquad
      g_2(\mathcal{Q}_{22}) =
      \stsp{A_K}{\bar L_2}{E_1^\tp(K_d-K)}{0}
    \end{aligned}
    \end{equation}
\fi
where $\bar K_2$, $\bar L_2$ are defined in Lemmas~\ref{lem:fixed_Q11} and~\ref{lem:fixed_Q22} respectively.
\end{thm}

\begin{proof}
This follows directly from Lemmas~\ref{lem:fixed_Q11}
and~\ref{lem:fixed_Q22} and some simple state-space manipulations.
Note that $g_2(\mathcal{Q}_{22})$ depends on $\mathcal{Q}_{22}$ only
through the realization-independent quantities $C_QA_Q^k B_Q$ for $k\geq 0$,
and similarly for $g_1$.
\end{proof}

\begin{proofe}{Proof of Theorem~\ref{thm:main}, Part (iii).}
  Applying the fixed-point results of
 Lemma~\ref{lem:coupled_problems}, there must exist matrices
 $A_P$, $B_P$, $C_P$, $A_Q$, $B_Q$, $C_Q$ such that
\begin{align}
\stsp{A_P}{B_P}{C_P}{0} &= \stsp{A_K}{\bar L_2}{E_1^\tp(K_d-K)}{0} \label{ff1}\\
\stsp{A_Q}{B_Q}{C_Q}{0} &= \stsp{A_L}{(L_d-L)E_2}{\bar K_2}{0} \label{ff2}
\end{align}
where \eqref{EQN:phi_final} and \eqref{EQN:psi_final} are
satisfied. Given $A_Q,B_Q,C_Q$ we define $\Psi$ according
to~\eqref{EQN:psi_final}, and given $A_P,B_P,C_P$ we define $\Phi$
according to~\eqref{EQN:phi_final}.  We will show that these $\Phi$
and $\Psi$ satisfy the Sylvester
equations~\eqref{tp:phi}--\eqref{tp:psi}.

For convenience, define $\mathcal{S} \defeq (
A_P,B_P,C_P,A_Q,B_Q,C_Q).$ Our goal is to use~\eqref{ff1}--\eqref{ff2}
to eliminate the matrices in $\mathcal{S}$ from \eqref{EQN:phi_final}
and \eqref{EQN:psi_final}. We begin with~\eqref{ff1}.

Note that we cannot simply set the corresponding state-space
parameters in~\eqref{ff1} equal to one another. This approach is
erroneous because transfer function equality does not in general imply
that the state-space matrices are also equal. For example, if $E_1^\tp
(K_d-K) = 0$, then we can set $C_P=0$ and any choice of $B_P$ satisfies~\eqref{ff1}. Equality of transfer functions does however imply equality of the Markov parameters. Namely,
\begin{equation}\label{markov_params}
C_P A_P^k B_P = E_1^\tp (K_d - K) A_K^k \bar L_2
\quad \text{for }k=0,1,\dots
\end{equation}
Now consider~\eqref{EQN:phi_final}. Note that $\Theta_X$ does not depend on $\mathcal{S}$. Furthermore, the equation for $\Gamma_X$ is a Sylvester equation of the form $A_J^\tp \Gamma_X + \Gamma_X A_P + \Omega C_P = 0$, where
\[
\Omega \defeq \bmat{ \Theta_X & \tilde X}B_2E_1 + J^\tp R_{21} + S_{21}
\]
and $\Omega$ is independent of $\mathcal{S}$.  Since $A_J$ and $A_P$
are Hurwitz by construction, the unique $\Gamma_X$ is given by the integral
\[
\Gamma_X = \int_0^\infty \exp(A_J^\tp t)\, \Omega\,  C_P \exp(A_Pt) \,\mathrm{d}t
\]
Substitute the Markov parameters~\eqref{markov_params}, and conclude that
\[
\Gamma_X B_P = \int_0^\infty \exp(A_J^\tp t)\, \Omega\,  E_1^\tp(K_d-K)\exp(A_K t) \bar L_2 \,\mathrm{d}t
\]
The equation above is of the form $\Gamma_X B_P = \hat{\Gamma}_X \bar
L_2$, where $\hat{\Gamma}_X$ satisfies
\begin{equation}
        A_J^\tp \hat{\Gamma}_X + \hat{\Gamma}_X A_K +
        \Omega E_1^\tp (K_d-K) = 0 \label{hh1}
\end{equation}
One can verify by direct substitution that~\eqref{hh1} is solved by
\begin{equation}\label{gamx}
\hat{\Gamma}_X = \bmat{ \Theta_X - X_{21} & \tilde X - X_{22} }
\end{equation}
Noting that the term $\Gamma_X B_P$ appears explicitly in the
$\Phi$-equation of~\eqref{EQN:phi_final}, we may replace the $\Gamma_X B_P$ term by $\hat{\Gamma}_X \bar L_2$, and substitute the expressions for
$\hat{\Gamma}_X$ and $\bar L_2$ from~\eqref{gamx} and
Lemma~\ref{lem:fixed_Q11}, respectively. Doing so, we find that
$\Theta_X$ cancels. The result is an equation involving only~$\Phi$
and~$\Psi$, and it turns out to be~\eqref{tp:phi}.

Repeating a similar procedure as above by instead
  examining~\eqref{ff2} and~\eqref{EQN:psi_final}, we find that
  $C_Q\Gamma_Y = \bar K_2 \hat{\Gamma}_Y$ where $\hat{\Gamma}_Y$ is
  given by
  \[
  \hat{\Gamma}_Y = \bmat{\tilde Y-Y_{11} \\ \Theta_Y - Y_{21}}
  \]
  The result is a different equation involving only~$\Phi$
  and $\Psi$. This time, it turns out to be~\eqref{tp:psi}.

We have therefore shown that existence of a solution to the
fixed-point equations of Lemma~\ref{lem:coupled_problems} implies the
existence of a solution to the Sylvester
equations \eqref{tp:phi}--\eqref{tp:psi}. This completes the
proof.
\end{proofe}


\section{Summary}

\label{sec:conclusion}

In this article, we studied the class of two-player output-feedback
problems with a nested information pattern.

We began by giving necessary and sufficient state-space conditions for the existence of a structured stabilizing controller. This led to a Youla-like parameterization of all such controllers and a convexification of the two-player problem.

The main result of this paper is explicit state-space formulae for
the optimal $\Htwo$ controller for the two-player output-feedback
problem.  In the centralized case, it is a celebrated and
widely-generalized result that the controller is a composition of an
optimal state-feedback gain with a Kalman filter estimator.  Our
approach generalizes both the centralized formulae and this separation
structure to the two-player decentralized case.  We show that the
$\Htwo$-optimal structured controller has generically twice the state dimension of the plant, and we give intuitive interpretations for the states of the controller as steady-state Kalman filters. The player with more
information must duplicate the estimator of the player with less
information. This has the simple anthropomorphic interpretation that
Player~2 is correcting mistakes made by Player~1.

Both the state-space dimension and separation structure of the optimal
controller were previously unknown. While these results show that the
optimal controller for this problem has an extremely simple
state-space structure, not all such decentralized problems exhibit
such pleasant behavior. One example is the two-player
fully-decentralized state-feedback problem, where even though the
optimal linear controller is known to be rational, it is shown
in~\cite{lessard2012fully} that the number of states of the
controller may grow quadratically with the state-dimension of the plant.

The formulae that we give for the optimal controller are simply
computable, requiring the solution of four standard AREs, two that have the same dimension as the plant and two
with a smaller dimension. In addition, one must solve a linear
matrix equation.  All of these computations are simple and have
readily available existing code.

While there is as yet no complete state-space theory for decentralized
control, in this work we provide solutions to a prototypical class of
problems which exemplify many of the features found in more general
problems.  It remains a fundamental and important problem to fully
understand the separation structure of optimal decentralized
controllers in the general case. While we solve the two-player
triangular case, we hope that the solution gives some insight and
possible hints regarding the currently unknown structure of the
optimal controllers for more general architectures.


\section*{Acknowledgments}

The first author would like to thank B.~Bernhardsson
for very helpful discussions.

\bibliographystyle{abbrv}
\bibliography{twoplayer}

\if\MODE2
\begin{IEEEbiography}[{\includegraphics[width=1in,height=1.25in,clip,keepaspectratio]{LL_500}}]{Laurent Lessard} was born and raised in Toronto, Canada. He received the BASc in Engineering Science at the University of Toronto and the M.S. and Ph.D. degrees in Aeronautics and Astronautics from Stanford University.
He is currently a postdoctoral scholar in the Berkeley Center for Control and Identification at the University of California, Berkeley. Before that, he was a LCCC postdoc in the Department of Automatic Control at Lund University. His research interests include decentralized control, robust control, and large-scale optimization.
Dr. Lessard received the O.~Hugo Schuck Best Paper Award at the American Control Conference in 2013.
\end{IEEEbiography}

\begin{IEEEbiography}[{\includegraphics[width=1in,height=1.25in,clip,keepaspectratio]{lall_cropped}}]{Sanjay Lall}
is Professor in the departments of Electrical Engineering and
Aeronautics and Astronautics at Stanford University.  Previously he
was a Research Fellow at the California Institute of Technology in the
Department of Control and Dynamical Systems, and prior to that he was
NATO Research Fellow at Massachusetts Institute of Technology, in the
Laboratory for Information and Decision Systems. He was also a
visiting scholar at Lund Institute of Technology in the Department of
Automatic Control.  He received the Ph.D. in Engineering and B.A. in
Mathematics from the University of Cambridge, England. Professor
Lall's research focuses on the development of advanced engineering
methodologies for the design of control systems, and his work
addresses problems including decentralized control and model
reduction.  Professor Lall received the O. Hugo Schuck Best Paper
Award at the American Control Conference in 2013, the George S. Axelby
Outstanding Paper Award by the IEEE Control Systems Society in 2007,
the NSF Career award in 2007, Presidential Early Career Award for
Scientists and Engineers (PECASE) in 2007, and the Graduate Service
Recognition Award from Stanford University in 2005. With his students,
he received the best student paper award at the IEEE Conference on
Decision and Control in 2005 and the best student paper award at the
IEEE International Conference on Power Systems Technology
in 2012.
\end{IEEEbiography}

\vfill
\fi


\end{document}